\documentclass[11pt,letterpaper]{article}

\usepackage[margin=1in]{geometry}
\usepackage{setspace}
\usepackage{graphicx}
\usepackage{amsmath,amssymb,amsthm}
\RequirePackage{tgtermes}
\RequirePackage{newtxtext}
\RequirePackage{newtxmath}
\RequirePackage{bm}
\usepackage{booktabs}
\usepackage{enumitem}
\usepackage{array}
\usepackage{multirow}
\usepackage{natbib}
\bibpunct[, ]{(}{)}{,}{a}{}{,}

\numberwithin{equation}{section}
\makeatletter
\newcommand{\RUNAUTHOR}[1]{}
\newcommand{\RUNTITLE}[1]{}
\newcommand{\TITLE}[1]{\gdef\theTITLE{#1}}
\newcommand{\ARTICLEAUTHORS}[1]{\gdef\theARTICLEAUTHORS{#1}}
\newcommand{\AUTHOR}[1]{\par\noindent\textbf{#1}\par}
\newcommand{\AFF}[1]{\noindent{\small #1}\par\vspace{0.55em}}
\newcommand{\EMAIL}[1]{\texttt{#1}}
\newcommand{\ABSTRACT}[1]{\gdef\theABSTRACT{#1}}
\newcommand{\KEYWORDS}[1]{\gdef\theKEYWORDS{#1}}
\newcommand{\ECSwitch}{\clearpage\appendix}
\newcommand{\ECHead}[1]{\section*{#1}\addcontentsline{toc}{section}{#1}}
\renewcommand{\maketitle}{%
  \begin{center}
    {\LARGE\bfseries \theTITLE\par}
  \end{center}
  \vspace{0.7em}
  \begin{minipage}{\textwidth}
    \theARTICLEAUTHORS
  \end{minipage}
  \vspace{0.8em}
  \begin{abstract}
    \theABSTRACT
  \end{abstract}
  \noindent\textbf{Keywords:} \theKEYWORDS\par\vspace{1em}
}
\makeatother
\theoremstyle{plain}
\newtheorem{theorem}{Theorem}[section]
\newtheorem{proposition}[theorem]{Proposition}

\newtheorem{corollary}[theorem]{Corollary}
\newtheorem{assumption}[theorem]{Assumption}
\newtheorem{definition}[theorem]{Definition}
\theoremstyle{remark}
\newtheorem{remark}[theorem]{Remark}

\newcommand{\R}{\mathbb{R}}
\newcommand{\E}{\mathbb{E}}
\newcommand{\Var}{\operatorname{Var}}
\newcommand{\tr}{\operatorname{tr}}
\DeclareMathOperator*{\argmin}{arg\,min}

\begin{document}

\RUNAUTHOR{Zhou et al.}
\RUNTITLE{A strictly proper scoring-rule theory for stochastic car-following models}
\TITLE{A Strictly Proper Scoring-Rule Theory for Calibrating Stochastic Car-Following Models}

\ARTICLEAUTHORS{%
\AUTHOR{Shirui Zhou}
\AFF{Institute of Systems Engineering, College of Management and Economics,
Tianjin University, Tianjin 300072, China; Laboratory of Computation and
Analytics of Complex Management Systems (CACMS), Tianjin University,
Tianjin 300072, China}
\AUTHOR{Shiteng Zheng}
\AFF{School of Systems Science, Beijing Jiaotong University, Beijing 100044, China}
\AUTHOR{Junzhe Ding}
\AFF{Institute of Systems Engineering, College of Management and Economics,
Tianjin University, Tianjin 300072, China; Laboratory of Computation and
Analytics of Complex Management Systems (CACMS), Tianjin University,
Tianjin 300072, China}
\AUTHOR{Rui Jiang}
\AFF{School of Systems Science, Beijing Jiaotong University, Beijing 100044, China}
\AUTHOR{Junfang Tian\textsuperscript{*}}
\AFF{Institute of Systems Engineering, College of Management and Economics,
Tianjin University, Tianjin 300072, China; Laboratory of Computation and
Analytics of Complex Management Systems (CACMS), Tianjin University,
Tianjin 300072, China, \EMAIL{jftian@tju.edu.cn}. \textsuperscript{*}Corresponding author.}
}

\ABSTRACT{%
\textbf{Problem definition:} Fixed parameters and inputs in a stochastic simulator induce a distribution over complete trajectories, not one trajectory. Calibration must assess this distribution, including variability and temporal dependence, against observations. Yet stochastic car-following models are commonly calibrated with trajectory-error objectives inherited from deterministic modelling. \textbf{Methodology/results:} We establish a scoring-rule theory of stochastic calibration. Strict propriety requires the data-generating distribution to uniquely minimise expected score. MRMean-I, the average run-wise error, drives separable stochastic spread to zero; MRMean-II, the error of the ensemble-mean trajectory, cannot identify a parameter that changes only spread; and MRMin, the error of the closest simulated run, has a population target that changes with ensemble size. These results are confirmed for stochastic Intelligent Driver Model extensions with additive acceleration noise and random desired headway. We recommend exact maximum likelihood when the correct transition density is available; otherwise, an unbiased simulation-based estimator of a strictly proper score. The energy score meets this requirement and gives the best held-out distributional prediction among the evaluated simulation-based objectives, although both models retain too-narrow bands and miss persistent disturbances. \textbf{Implications:} Strict propriety separates a valid calibration target from parameter identifiability and model adequacy. The theory applies to vector-valued outputs from stochastic transportation simulators; the car-following experiments illustrate its scope.
}


\maketitle

\section{Introduction}
\label{sec:introduction}

Traffic simulation supports the design and evaluation of vehicle control, infrastructure, and traffic-management policy when field experiments are costly, unsafe, or impossible \citep{punzo2021about}. Transportation Science has contributed calibration methods for dynamic traffic simulators and car-following models using both aggregate and trajectory data \citep{flotterod2011bayesian,rakha2009calibrating,keane2021fast}. A deterministic simulator maps a parameter vector to one predicted outcome. A stochastic simulator instead maps the same vector to a predictive distribution $P_\beta$: repeated runs under identical inputs can produce different trajectories, link states, or network outcomes \citep{osorio2019efficient,punzo2020twolevel}. The distribution is the model's prediction. It includes not only a mean path, but also variation, tails, and temporal or cross-component dependence. An observed path is one realisation from the unknown data-generating distribution $P_0$, not a direct observation of all of its features.

Consequently, calibrating a stochastic simulator is not the deterministic task of matching one predicted trajectory to one observed trajectory. It is the task of selecting parameters for which $P_\beta$ is an appropriate predictive distribution for outcomes drawn from $P_0$. Two parameter vectors can give the same mean trajectory while assigning different probabilities to large gaps, stop-and-go waves, or persistent accelerations. A point-error criterion cannot distinguish them if it discards those features. This distinction matters because safety depends on distributional tails \citep{mahmud2017application}, congestion depends on the propagation of variation \citep{laval2014parsimonious,tian2016empirical}, and emissions estimates depend on how acceleration variability evolves through a platoon \citep{darocha2015emissions,zhou2023emissions}. A calibration objective can therefore fit the mean and converge numerically while still favouring a predictive distribution that is too narrow or otherwise wrong.

Stochastic car-following (CF) models make this problem especially transparent and especially difficult. CF is a basic microscopic component of traffic simulation \citep{brackstone1999historical}, but its trajectory is generated recursively: each simulated state depends on earlier simulated states. A random acceleration or headway innovation therefore propagates through time, creating a high-dimensional distribution with temporal dependence rather than independent pointwise errors \citep{punzo2016speed}. In field data, this within-driver run-to-run variability is also observed through measurement error and alongside stable heterogeneity between drivers. Moreover, a stochastic parameter may change only predictive spread, or may change mean and spread together. Matching one realised path can thus confound model adequacy, parameter information, and the rule used to compare a simulated ensemble with data. Recent Transportation Science work similarly notes that validation of a stochastic microscopic model cannot rely on either an arbitrary parameter realisation or the best-matching simulated trajectory \citep{zhao2022microscopic}. Existing stochastic CF mechanisms span noise-driven IDM variants, random-headway models, two-regime dynamics, and Langevin formulations \citep{treiber2017idmstochasticity,zhou2025calibration,xu2019tworegime,ngoduy2019langevin}, but their multi-run calibration objectives have not been assessed against a common distributional validity criterion.

The missing ingredient is a criterion for judging a predictive distribution against a realised outcome. We formulate calibration as the choice of a \emph{scoring rule}: given a candidate distribution and an observation, the rule assigns a penalty, and calibration selects the parameter vector with the smallest penalty. A scoring rule is \emph{strictly proper} when, over indefinitely repeated observations from $P_0$, the true distribution uniquely attains the smallest expected penalty \citep{gneiting2007strictly}. Strict propriety is therefore the universal standard for distributional calibration: it asks whether a criterion targets the correct predictive distribution independently of the optimiser, a particular dataset, or a numerical choice such as the number of simulator runs. It does not guarantee that the simulator can represent $P_0$, nor that different parameter vectors can be distinguished when they generate nearly identical distributions; those are model-adequacy and identifiability questions, respectively.

This standard lets us derive the failures of three objectives used in recent stochastic CF calibration \citep{zhou2025calibration}. MRMean-I averages the errors of individual simulation runs; MRMean-II calculates the error of their average trajectory; and MRMin retains the error of the closest run. All appear to be harmless ways to reduce an ensemble to a scalar objective, but they score different aspects of $P_\beta$. We show that MRMean-I penalises separable stochastic spread, MRMean-II is uninformative about a parameter that changes only spread, and MRMin changes its population objective and minimiser when the ensemble size $N$ changes. The first two methods are therefore not merely noisy estimators of a common target, and a favourable fixed-$N$ parameter-recovery exercise does not establish MRMin as a general rule.

The appropriate alternative depends on what the simulator makes available. If its state-transition density can be evaluated, exact maximum likelihood estimation (MLE) directly uses the probability law that generates the data \citep{xu2020statistical,ngoduy2019langevin}. This is distinct from an approximate residual likelihood that posits, for example, independent Gaussian trajectory errors; poor performance of the latter is evidence against that residual model, not against exact MLE. If the simulator is a black-box generator and its output density is unavailable, the calibration objective should be an unbiased estimator of a proper score. We use the multivariate energy score, whose pairwise ensemble term rewards spread only when it improves the distributional prediction \citep{szekely2013energy,gneiting2008assessing}. Section~\ref{sec:theory} gives the formal definitions and population calculations.

\subsection{From deterministic trajectory fitting to distributional calibration}
\label{sec:related-work}
\label{sec:rw-deterministic}
\label{sec:rw-stochastic}
\label{sec:rw-diagnostics}
\label{sec:rw-scoring}
\label{sec:rw-sbi}

Deterministic CF calibration has established the importance of the measured variable, data quality, numerical implementation, parameter sensitivity, and optimisation \citep{punzo2016speed,punzo2021about,rakha2009calibrating,keane2021fast}. Those results remain essential, but their usual point-to-point loss presumes that a parameter vector determines one trajectory. Stochastic CF mechanisms were introduced to reproduce variability and oscillation phenomena that a deterministic path cannot describe \citep{treiber2017idmstochasticity,tian2019speedadaptation,xu2019tworegime, ngoduy2019langevin}. Once a model generates an ensemble, however, the measure of performance and optimiser no longer determine the estimand: averaging run-wise errors, scoring the ensemble mean, and retaining the closest run are different implicit scoring rules. Table~\ref{tab:rw-implicit-scores} records this reclassification.

The closest traffic-specific precursor is \citet{zhou2025calibration}, which showed that averaging run-wise errors can suppress stochasticity and reported favourable fixed-$N$ recovery for MRMin. We retain that empirical starting point but change the question from recovery on one finite design to the population objective and minimiser of each rule. This separates the two MRMean definitions and shows when a reported parameter can be an artefact of the simulation ensemble rather than an implication of the data. Table~\ref{tab:progression} summarises the distinction.

This framing connects stochastic traffic calibration to probabilistic forecast evaluation, where an ensemble is routinely assessed against a realised outcome using proper scoring rules \citep{gneiting2007probabilistic,gneiting2008assessing, jordan2019scoringrules}. It also clarifies validation. Held-out point error may improve simply because a model contracts its predictive distribution, so repeated observations are needed to distinguish within-driver stochastic variation from stable between-driver heterogeneity \citep{punzo2020twolevel}. We use the variogram score as a complementary diagnostic for temporal dependence, which the energy score alone is less sensitive to \citep{scheuerer2015variogram}. The same information principle is consistent with simulation-based inference: likelihood is preferred when it is available; simulation-only methods still need a well-defined population objective \citep{cranmer2020frontier}. In particular, increasing $N$ should reduce Monte Carlo error, not change the distribution that calibration is intended to select.

The paper makes four contributions. First, it establishes strict propriety as the universal validity criterion for simulation-based distributional calibration. Second, it derives, at the population level, why MRMean-I, MRMean-II, and MRMin fail that criterion in three different ways. Third, it turns the theory into a decision rule: use exact MLE when the correct transition density is available and otherwise use an unbiased estimator of a proper simulation-based score. The energy score is an existing proper score; our contribution is to establish its role in this calibration framework. Fourth, controlled experiments and repeated trajectories from 11 drivers test the theoretical signatures, held-out distributional prediction, and the temporal features that the fitted stochastic mechanisms still miss.

The theoretical statements concern distributions over vector-valued simulator outputs, not one CF equation. They apply whenever repeated runs at a parameter vector define a predictive distribution over a finite-dimensional output, such as a trajectory, link state, origin--destination cell, or safety measure. QIDM, an Intelligent Driver Model (IDM) with additive Gaussian white-noise acceleration of intensity $Q$, and the two-dimensional IDM (2D-IDM), which randomises desired time headway, are two canonical but structurally distinct stochastic mechanisms: the former perturbs instantaneous acceleration, whereas the latter randomises a behavioural control variable. They are controlled validation vehicles rather than boundaries on the theory's scope.

The remainder of the paper is organised as follows. Section~\ref{sec:theory} develops the scoring-rule theory and the energy-score alternative. Section~\ref{sec:synthetic} tests its analytic signatures under controlled synthetic truths. Section~\ref{sec:experimental} uses repeated trajectories for predictive validation, comparison of stochastic signatures, and parameter uncertainty. Section~\ref{sec:conclusions} gives practical guidance and limitations.

\section{A scoring-rule theory of stochastic calibration}
\label{sec:theory}

\subsection{Notation and setup}
\label{sec:theory-notation}

For the application in this paper a measure-of-performance (MoP) is spacing, and a
single realisation is a trajectory $y \in \R^K$, where $K$ is the number of discrete
time steps. The theory uses only that a transportation simulator produces a predictive
distribution over a vector-valued output: the components may instead index links,
vehicles, origin--destination cells, performance measures, or their concatenation.
We write $P_\beta$ for the distribution induced by a stochastic simulator
at parameter vector $\beta$, with mean $\mu_\beta = \E_{X \sim P_\beta}[X] \in \R^K$
and covariance $\Sigma_\beta = \Var_{X \sim P_\beta}(X) \in \R^{K \times K}$. We write
$P_0$ for the (unknown) true data-generating distribution, with mean $\mu_0$ and
covariance $\Sigma_0$. An observed trajectory is $y \sim P_0$. We write
$X, X' \stackrel{\mathrm{iid}}{\sim} P_\beta$, both independent of $y$, and
$X_1,\dots,X_N \stackrel{\mathrm{iid}}{\sim} P_\beta$ for a set of $N$ independent
simulation runs at parameter $\beta$.

Here and throughout, \emph{population} does not mean the 11 drivers. It means the
probability distribution revealed by repeating the same experiment indefinitely.
A population objective therefore averages over a new trajectory $Y\sim P_0$ and,
where needed, new simulator randomness. Increasing the number $N$ of simulated runs
for one observed trajectory does not create this population \citep{newey1994large}.

Three parameter values must be kept distinct. If the model is correctly
specified, $\beta_0$ denotes a data-generating parameter satisfying
$P_0=P_{\beta_0}$. For a scoring rule $S$, the population minimiser is
\[
  \beta_S^*\in\argmin_\beta \E_{Y\sim P_0}S(P_\beta,Y),
\]
or $\beta_{S,N}^*$ when the population objective itself depends on ensemble size.
Finally, $\widehat\beta_{n,N}$ is the value obtained from the available data: $n$
observed trajectories and $N$ simulations per score evaluation. We call
$\beta_S^*$ (or $\beta_{S,N}^*$) the \emph{estimand}, meaning the value that the
population objective selects on average. If this minimiser changes with $N$, the criterion itself
has changed its answer. Random variation of $\widehat\beta_{n,N}$ around a fixed
target is instead ordinary estimation error. Neither changes the true value
$\beta_0$.

Throughout, goodness-of-fit (GoF) is mean squared error, i.e.\ the loss is built from
$\|\cdot\|^2$, the squared Euclidean norm on $\R^K$. We state this convention once,
explicitly, because every impropriety result in Section~\ref{sec:theory-improper}
is a direct consequence of this specific choice, and Remark~\ref{rem:gof-generalises}
addresses what changes if a different GoF is substituted.

\subsection{Calibration as the choice of a scoring rule}
\label{sec:theory-scoring}

The traditional description of CF calibration proceeds by fixing a measure of
performance (typically spacing or speed) and a goodness-of-fit statistic, then
choosing $\beta$ to minimise the discrepancy between a simulated trajectory and an
observed one. This ``MoP + GoF'' decomposition is standard for deterministic CF
models, where a single parameter vector $\beta$ produces a single simulated
trajectory and the comparison is genuinely point-to-point. For a stochastic CF
model, however, $\beta$ does not produce a trajectory; it produces a distribution
$P_\beta$ over trajectories. The MoP + GoF decomposition presumes a point-to-point
comparison that a stochastic model does not, by construction, supply, and any
procedure that forces one (for example, by comparing $y$ to a single simulated
draw, to the mean of several draws, or to the closest of several draws) is making an
implicit and consequential choice about how a distribution is to be scored against
a realisation.

We make this choice explicit. Calibration is the problem
\begin{equation}
  \hat\beta = \argmin_\beta S(P_\beta, y),
  \label{eq:calibration-as-scoring}
\end{equation}
where $S(P, y)$ is a \emph{scoring rule}: a function assigning a real-valued
penalty to a predictive distribution $P$ given a realised observation $y$. Every
calibration method in current use for stochastic CF models is an instance of
\eqref{eq:calibration-as-scoring} for a particular, usually implicit, choice of
$S$. Section~\ref{sec:theory-improper} identifies the $S$ implicit in each of the
three examined methods and analyses their behaviour; Section~\ref{sec:theory-es}
proposes a replacement.

One observed trajectory can favour an incorrect distribution by chance. The relevant
question is therefore what happens on average over repeated observations. A valid
calibration score must give the true distribution no larger an expected penalty than
any alternative. Strict propriety further requires every incorrect alternative to
have a larger expected penalty.

\begin{definition}[Propriety]
A scoring rule $S$ is \emph{proper} relative to a class $\mathcal{P}$ of predictive
distributions if, for every $P_0 \in \mathcal{P}$,
\[
  \E_{y \sim P_0}\, S(P_0, y) \;\le\; \E_{y \sim P_0}\, S(P, y)
  \qquad \text{for all } P \in \mathcal{P}.
\]
It is \emph{strictly proper} if the inequality is strict whenever $P \ne P_0$.
\end{definition}

Strict propriety is precisely the condition needed for \eqref{eq:calibration-as-scoring}
to recover $P_0$ as the unique population-optimal \emph{distribution}. In parameter
space, the population argmin is the equivalence class
$\{\beta:P_\beta=P_0\}$; it reduces to the singleton $\{\beta_0\}$ only when the
simulator family is identifiable. A scoring rule that is proper but not strict admits
minimisers other than $P_0$; a scoring rule that is not proper at all can
prefer a wrong distribution to the true one in expectation, regardless of
how much data is available. The thesis of this paper is that \emph{the
decisive criterion for distributional calibration is strict propriety, and
the three examined methods fail it in different ways}: one by an
impropriety that drives the dispersion parameters to a boundary
(Section~\ref{sec:mrmean1}), one by an impropriety that fails to identify
them at all (Section~\ref{sec:mrmean2}), and one by an impropriety that
makes the calibrated estimate depend on a nuisance quantity --- the number of
simulation runs, $N$ --- that has no counterpart in the estimand
(Section~\ref{sec:mrmin}).

\begin{remark}[Propriety is not identifiability]
\label{rem:propriety-not-identifiability}
Strict propriety says which \emph{distribution} wins on average: if the true
distribution $P_0$ is among the candidates, it is the unique minimiser. It does not
guarantee that a finite dataset can distinguish every \emph{parameter}. Two nearby
parameter values may still generate almost indistinguishable trajectories. In
Section~\ref{sec:exp8} we examine this using the curvature of the fitted objective,
that is, how quickly the score worsens as a parameter moves away from its optimum.
A flat objective means weak parameter information, not a failure of propriety.
\end{remark}

Table~\ref{tab:criteria} condenses the distinction that drives the analysis.
The first two rules treat dispersion asymmetrically or omit it, while MRMin
makes it depend on ensemble size. Only the energy score balances fidelity to
the realised path against a reward for justified ensemble spread.

\begin{table*}[t]
\centering
\caption{Population objectives and identifying properties of the four calibration
criteria. The table is an analytic summary of Section~\ref{sec:theory}; $N$ is the
number of simulated trajectories and is a tuning parameter only for MRMin.}
\label{tab:criteria}
\small
\begin{tabular}{@{}p{0.13\linewidth}p{0.34\linewidth}p{0.15\linewidth}p{0.11\linewidth}p{0.10\linewidth}@{}}
\toprule
Criterion & Population objective & Dispersion term & Strictly proper & Tuning parameter \\
\midrule
MRMean-I & $\|\mu_\beta-\mu_0\|^2+\tr\Sigma_\beta+\tr\Sigma_0$ & penalised & no & none \\
MRMean-II & $\|\mu_\beta-\mu_0\|^2+\tr\Sigma_0$ & absent & no & none \\
MRMin & nearest-neighbour objective; $N$-indexed, with Gaussian-toy limit
$\sigma_0\sqrt{1+2/K}$ & implicit through $N$ & no & $N$ \\
Energy score & $\E\|X-y\|-\tfrac12\E\|X-X'\|$ & rewarded & yes & none \\
\bottomrule
\end{tabular}
\end{table*}

\begin{assumption}
\label{assump:standing}
The following conditions are in force throughout Section~\ref{sec:theory-improper}
unless stated otherwise:
\begin{enumerate}[label=(\alph*)]
  \item $P_0$ and $P_\beta$ have finite second moments, i.e.\ $\tr\Sigma_0 < \infty$
        and $\tr\Sigma_\beta < \infty$ for all $\beta$ in the parameter space.
  \item $X \perp y$: simulated trajectories are drawn independently of the observed
        trajectory.
  \item \emph{(Separability, used only where explicitly invoked.)} The mean and
        dispersion structures of $P_\beta$ are separately parameterisable: there
        exists $\beta$ with $\mu_\beta = \mu_0$ while $\Sigma_\beta$ ranges, on that
        same mean-preserving slice, over an interval (or region) containing values
        arbitrarily close to $0$ and values bounded away from it.
\end{enumerate}
Condition (c) is invoked only in Theorem~\ref{thm:mrmean1} and is flagged there
explicitly. Section~\ref{sec:mrmean2} discusses, as a mandatory caveat rather than
a standing assumption, what happens when (c) fails because a dispersion parameter
also shifts the mean.
\end{assumption}

\subsection{Impropriety of current calibration methods}
\label{sec:theory-improper}

\subsubsection{MRMean-I: variance collapse}
\label{sec:mrmean1}

MRMean-I applies the error first and the averaging second. For a finite ensemble
$X_1,\ldots,X_N$, it uses
\[
  \widehat S_{1,N}(y)=\frac{1}{N}\sum_{i=1}^N\|X_i-y\|^2,
\]
whose population version is the expected squared distance from a single simulated
draw to the observation:
\[
  S_1(P_\beta, y) = \E_X \|X - y\|^2 .
\]

\begin{theorem}[Variance collapse]
\label{thm:mrmean1}
Under Assumption~\ref{assump:standing}(a)--(b),
\begin{equation}
  \E_{y}\, S_1(P_\beta, y)
  = \|\mu_\beta - \mu_0\|^2 + \tr\Sigma_\beta + \tr\Sigma_0 .
  \label{eq:s1-decomp}
\end{equation}
Consequently, under the additional separability condition,
Assumption~\ref{assump:standing}(c), the infimum is attained at zero dispersion
when that boundary is admissible; otherwise every minimising sequence satisfies
$\tr\Sigma_\beta\to0$ while holding $\mu_\beta=\mu_0$. Thus the stochastic
parameters are driven to the lower boundary of their admissible range and the
calibrated model approaches its deterministic core. $S_1$ is not proper.
\end{theorem}

In plain language, MRMean-I counts every difference between a simulated run and the
observation as error. It therefore improves its average score by removing simulated
run-to-run spread, even when that spread is real.

\noindent\textit{Proof outline.} The decomposition follows by centring the simulated and observed trajectories: all cross terms have zero expectation, leaving mean mismatch plus the two covariance traces. On a mean-preserving parameter slice, reducing predictive dispersion strictly reduces MRMean-I. The full proof is in the Electronic Companion.

\begin{remark}
No distributional assumption beyond finite second moments and independence
(Assumption~\ref{assump:standing}(a)--(b)) was used in the squared-loss identity.
The result therefore applies to any stochastic transportation simulator whose
vector-valued output satisfies those conditions and whose mean and dispersion can
be varied as required by Assumption~\ref{assump:standing}(c). QIDM and 2D-IDM are
validation vehicles for that score-level result, not boundaries on its scope.
\end{remark}

\begin{remark}[Run-wise RMSE and other GoF]
\label{rem:gof-generalises}
Two superficially similar RMSE objectives must be distinguished.
$\sqrt{\E\|X-y\|^2}$ is a strictly increasing transform of $S_1$ and therefore
has the same minimiser.  The form used in practice is instead
\[
 S_1^{\rm RMSE}(P,y)=\E\|X-y\|,
\]
the mean of run-wise RMSEs, as in Equation~(4) and the MRMean entries of
Table~1 in \citet{zhou2025calibration}.  Jensen's inequality compares these two
objectives but does not make their argmins identical.

The practical form is exactly the first term of the energy score; it lacks only
$-\tfrac12\E\|X-X'\|$.  It is likewise improper.  In the correctly centred
Gaussian toy,
\[
 \E\|X-Y\|=c_K\sqrt{\sigma^2+\sigma_0^2},
 \qquad c_K=\E\chi_K,
\]
which is strictly increasing in $\sigma$ and is minimised at the lower scale
boundary. Within the same correctly centred Gaussian scale family, replacing the
norm by a power $\|\cdot\|^p$, $p>0$, leaves its expectation proportional to
$(\sigma^2+\sigma_0^2)^{p/2}$ and again collapses at the lower boundary.
Thus collapse is not an artefact of squaring in this scale family. This calculation
is not asserted as a theorem for every convex loss and every predictive family:
outside the squared-loss identity of Theorem~\ref{thm:mrmean1}, the effect depends on
the loss, distribution family, and coupling of mean and dispersion.
Figure~\ref{fig:exp1-scores}(e) shows the Gaussian objectives side by side.
\end{remark}

\begin{remark}[Relation to prior work]
\citet{zhou2025calibration} established a variance-collapse result for MRMean
empirically and theoretically for specific stochastic CF models, via a
run-averaged argument tied to those models' functional forms. For squared Euclidean
loss, Theorem~\ref{thm:mrmean1} separates the result from those functional forms and
derives it for any predictive distribution with finite second moments, subject to
the stated independence and separability conditions. Remark~\ref{rem:gof-generalises}
shows the same direction for norm powers in the centred Gaussian scale family but
does not claim an arbitrary-loss theorem. The score-level squared-loss result is a
strengthening of that earlier analysis, not a correction of it.
\end{remark}


\subsubsection{MRMean-II: non-identification}
\label{sec:mrmean2}

MRMean-II reverses the operations: it averages the simulated trajectories first and
then computes one error,
\[
  \widehat S_{2,N}(y)=\left\|\frac{1}{N}\sum_{i=1}^N X_i-y\right\|^2.
\]
Its population version is the squared distance from the predictive mean to the
observation,
\[
  S_2(P_\beta,y)=\|\E_X X-y\|^2=\|\mu_\beta-y\|^2.
\]

\begin{theorem}[Non-identification]
\label{thm:mrmean2}
Under Assumption~\ref{assump:standing}(a)--(b),
\begin{equation}
\begin{aligned}
  \E_y S_2(P_\beta,y)
  &=\|\mu_\beta-\mu_0\|^2+\tr\Sigma_0,\\
  \frac{\partial}{\partial\Sigma_\beta}\E_yS_2(P_\beta,y)&\equiv0.
\end{aligned}
  \label{eq:s2-decomp}
\end{equation}
Thus a parameter that moves only $\Sigma_\beta$ is not identified by $S_2$.
\end{theorem}

In plain language, MRMean-II sees the mean trajectory but is blind to changes that
only widen or narrow the simulated trajectory distribution.

\noindent\textit{Proof outline.} After replacing the ensemble by its mean, the expected squared error contains only mean mismatch and the observation variance. The predictive covariance therefore disappears from the population objective. The full proof is in the Electronic Companion.

The contrast with Theorem~\ref{thm:mrmean1} is qualitative.  MRMean-I creates a
directed collapse and hence a repeatable boundary estimate.  MRMean-II contains
no population information about a pure dispersion coordinate, so finite-ensemble
noise and parameter bounds determine its estimate.  Section~\ref{sec:exp1}
provides numerical shadows of both statements.

\begin{remark}[Partial identification through a mean channel]
\label{rem:partial-identification}
If a nominal dispersion parameter also changes $\mu_\beta$, it enters
\eqref{eq:s2-decomp} through the location term.  MRMean-II may then identify the
parameter as a location effect, but it still receives no information from the
forecast-dispersion term $\Sigma_\beta$.
\end{remark}

\paragraph{Example 1 (explicit 2D-IDM bias).}
\label{ex:idm2d-bias}
In 2D-IDM,
\[
  T_n\sim U(T_{\min},T_{\min}+\Delta T),
  \qquad \E T_n=T_{\min}+\frac{\Delta T}{2}.
\]
Let the data-generating parameters be $(T_{\min,0},\Delta T_0)$, but hold the
fitted lower endpoint at $T_{\min,0}+b$.  Population MRMean-II matches the mean
channel, so
\[
  T_{\min,0}+b+\frac{\widehat{\Delta T}}{2}
  =T_{\min,0}+\frac{\Delta T_0}{2},
  \qquad
  \widehat{\Delta T}=\Delta T_0-2b,
\]
subject to the parameter bounds.  The parameter is identified, but as a
\emph{location} rather than a dispersion parameter: its bias is $-2b$ and is
unrelated to the true dispersion.  In particular, if $T_{\min}$ is held at its
true value ($b=0$), then $\widehat{\Delta T}=\Delta T_0$; agreement at the truth
cannot discriminate among criteria.  Section~\ref{sec:exp2} uses this prediction
as a design diagnostic.

\subsubsection{MRMin: dependence on the number of runs}
\label{sec:mrmin}

MRMin retains only the closest of $N$ simulated trajectories,
\[
  S_3^{(N)}(P_\beta,y)=\min_{i\leq N}\|X_i-y\|^2.
\]
It is therefore a nearest-neighbour criterion whose population minimiser can move
with the computational choice $N$.

For the isotropic Gaussian toy, let
$X_i\stackrel{\mathrm{iid}}{\sim}N_K(\mu_0,\sigma^2I)$ and
$Y\sim N_K(\mu_0,\sigma_0^2I)$ independently, and define
\[
\begin{aligned}
  g_N(\sigma)&=\E\min_{i\leq N}\|X_i-Y\|^2,\\
  \underline\sigma_N
    &=\inf\arg\min_{\sigma\geq0}g_N(\sigma).
\end{aligned}
\]
Conditioning on $Y=y$ gives the exact representation
\begin{equation}
  g_N(\sigma)=\E_Y\int_0^\infty
  2m\{1-p_{m,\sigma}(Y)\}^{N}\,dm,
  \label{eq:gN}
\end{equation}
where $p_{m,\sigma}(y)=\Pr_\sigma(\|X-y\|\leq m)$.

\begin{theorem}[Finite-$N$ comparative statics]
\label{thm:mrmin}
Put $h_N(u)=g_N(e^u)$.  If, on the scale interval under consideration,
\[
  h_{N+1}(u)-h_N(u)
\]
is strictly decreasing in $u$, then $\underline\sigma_{N+1}\geq
\underline\sigma_N$.  Thus the lower MRMin minimiser is non-decreasing wherever
the Gaussian objective has decreasing differences.  Because $g_N$ itself
depends on $N$, MRMin targets an $N$-indexed estimand rather than using $N$ only
to control Monte Carlo precision.
\end{theorem}

In plain language, adding simulations can systematically change the scale preferred
by MRMin. Here $N$ changes the calibration target, not merely the numerical accuracy
with which a fixed target is computed.

\noindent\textit{Proof outline.} The result is a decreasing-differences comparative-static argument: if the increment from $N$ to $N+1$ falls with log scale, the lower minimiser cannot move left. The full proof is in the Electronic Companion.

The decreasing-differences premise is a property of the Gaussian objective,
not an assumption on the CF model.  We verify it numerically rather than
establishing it analytically for every $(K,\sigma)$.  The theorem should
therefore be read conditionally: wherever the Gaussian objective has decreasing
differences in log scale, the MRMin minimiser moves outward with $N$.
Section~\ref{sec:exp1} supports the premise throughout the tested grid for
$K\geq10$, while also locating scalar regions where this sufficient condition
fails; monotonicity of the computed scalar minimisers does not turn that
sufficient condition into a global theorem.
For a new model, a practical finite-range diagnostic is to evaluate
$h_{N+1}(u)-h_N(u)$ on a rectangular grid in log scale around the fitted region
and check whether it decreases in $u$ at every tested $N$. Passing this check
supports the comparative-static direction only on that grid; failing it requires
reporting the empirical $N$ sensitivity without a monotonicity claim. An analytic
characterisation for general nonlinear transportation simulators remains future work.

\begin{remark}[Why the random seed is not a parameter]
\label{rem:seed-not-parameter}
The nested-optimisation interpretation in \citet{zhou2025calibration} treated the
random seed as an implicit variable and selected the seed whose trajectory best
matched the observation. That is a valid description of the computation but not of
the statistical estimand. Write $X_i=G(\beta,R_i)$ for simulator $G$ and independent
seeds $R_i$. At a fixed $\beta$ and observation $y$, if
$F_{\beta,y}(t)=\Pr\{\|X-y\|^2\leq t\}$, then
\[
 \Pr\!\left\{\min_{i\leq N}\|X_i-y\|^2>t\mid y\right\}
 =\{1-F_{\beta,y}(t)\}^{N}.
\]
The selected seed therefore supplies an order statistic whose distribution changes
with $N$; it is not a persistent feature of $P_\beta$ to be learned from the data.
Searching more seeds makes a rare chance match increasingly available and can favour
a different $P_\beta$. Calibration should integrate over simulator randomness when
scoring $P_\beta$, whereas optimising over the seed changes the target into a
best-of-$N$ nearest-neighbour criterion. The earlier fixed-$N$ recovery result remains
valid for its design, but it does not provide a population justification for MRMin.
\end{remark}

The pointwise limits in \eqref{eq:gN} require care.  At $\sigma=0$,
$g_N(0)=K\sigma_0^2$ for every $N$, and for each fixed $N$ the same value is
obtained as $\sigma\downarrow0$.  Conversely, for every fixed $\sigma>0$,
$p_{m,\sigma}(Y)>0$ almost surely for each $m>0$, so
$\{1-p_{m,\sigma}(Y)\}^N\to0$ and $g_N(\sigma)\to0$.  These are pointwise,
non-uniform limits: the second statement precludes convergence to the
small-scale floor uniformly in $N$.  It therefore does not imply that the
minimiser diverges.

\paragraph{High-dimensional finite-budget approximation.}
When $\log N=o(K)$, squared Gaussian distances have relative fluctuations of
order $K^{-1/2}$, and the minimum of $N$ approximately Gaussian distances probes
their left tail at a standardised displacement of order $\sqrt{\log N}$. This
suggests a square-root-logarithmic finite-budget approximation for the movement
of $\underline\sigma_N$; it does not by itself prove the rate of the minimiser.
Section~\ref{sec:exp1} treats this approximation as an empirical scaling law and
compares it with alternatives over the reported $(K,N)$ grid.

\begin{proposition}[Fixed-dimensional nearest-neighbour limit]
\label{prop:mrmin-dimension}
Fix $K$. Let $f_\sigma$ be the forecast density and $V_K$ the volume of the
$K$-dimensional unit ball. Assume that, for $P_0$-almost every $y$,
$f_\sigma$ is positive and continuous in a neighbourhood of $y$, so that
$P_\sigma\{B(y,r)\}=f_\sigma(y)V_Kr^K+o(r^K)$ as $r\downarrow0$, and assume
uniform integrability sufficient to average the conditional second-moment limit
over $Y$. Then, for fixed $K$ as $N\to\infty$,
\begin{equation}
  g_N(\sigma)\sim
  \frac{\Gamma(1+2/K)}{(NV_K)^{2/K}}
  \E\{f_\sigma(Y)^{-2/K}\}.
  \label{eq:mrmin-nn-asymptotic}
\end{equation}
For the Gaussian toy the expectation is finite when
$K\sigma^2>2\sigma_0^2$, and the scale factor in
\eqref{eq:mrmin-nn-asymptotic} is minimised at
\begin{equation}
  \sigma_\infty=\sigma_0\sqrt{1+2/K}.
  \label{eq:mrmin-finite-limit}
\end{equation}
Thus the ultimate fixed-dimensional nearest-neighbour target saturates at a
finite scale. The $\sqrt{\log N}$ description above concerns a separate,
high-dimensional finite-budget approximation.
\end{proposition}

In plain language, with a fixed trajectory dimension and extremely large $N$, the
best-of-$N$ scale approaches a finite value above the truth. The growth seen over
practical finite budgets is an earlier, high-dimensional regime, so the two results
describe different limits rather than contradicting one another.

\noindent\textit{Proof outline.} Conditional nearest-neighbour radii converge after the usual local-density rescaling. Averaging that limit and minimising its Gaussian scale factor gives $\sigma_0\sqrt{1+2/K}$. The full derivation is in the Electronic Companion.

For trajectories, $K$ is typically of order $10^3$ while $N$ is limited to
$10^2$--$10^3$, so $\log N\ll K$ and the Gaussian-left-tail regime is the
relevant one.  The distinction is material: at $K=1200$ and $N=2000$, the bare
nearest-neighbour radius factor $N^{1/K}-1$ is only $0.6\%$, whereas the observed
finite-budget movement exceeds $130\%$; Section~\ref{sec:exp1} reports the
goodness-of-fit comparison that selects the $\sqrt{\log N}$ form.

\begin{remark}[The small-scale floor is not uniform in $N$]
The equality $g_N(0)=K\sigma_0^2$ is independent of $N$, but no neighbourhood
of zero supplies an $N$-uniform floor: for each $\sigma>0$, increasing $N$
eventually drives $g_N(\sigma)$ to zero.  Consequently increasing $N$ moves the
finite-budget MRMin target but does not force it to infinity.  The energy score
in Section~\ref{sec:theory-es} instead has an $N$-invariant population minimiser at $\sigma_0$.
\end{remark}

\subsection{A strictly proper alternative: the energy score}
\label{sec:theory-es}

Define, for a predictive distribution $P$ with finite first moment and an
observation $y$,
\[
  \mathrm{ES}(P,y) = \E\|X-y\| - \tfrac12\, \E\|X-X'\|,
  \qquad X, X' \stackrel{\mathrm{iid}}{\sim} P .
\]
The energy score is strictly proper relative to the class of distributions on
$\R^K$ with finite first moment \citep{gneiting2007strictly}; we use this fact
and do not reprove it here.

\paragraph{Interpretation.} The second term, $-\tfrac12\E\|X-X'\|$, is a
\emph{dispersion reward}: it grows in magnitude as $P$ spreads out, and enters
the score with a negative sign, so that spreading the predictive distribution
lowers $\mathrm{ES}$ up to the point where the gain is no longer worth the
resulting mismatch with $y$ in the first term. Compare this directly with
Theorem~\ref{thm:mrmean1}: $S_1$ penalises $\tr\Sigma_\beta$ with nothing in the
score to offset it, which is exactly why its minimiser drives $\Sigma_\beta$ to
the boundary. The energy score contains exactly the compensating term that $S_1$
lacks, and it is this term --- not a different choice of MoP, and not a
different choice of GoF exponent (Remark~\ref{rem:gof-generalises}) --- that
repairs the method. The following results formalise this distinction.

\begin{proposition}[Unbiasedness in $N$]
\label{prop:es-unbiased}
Given $X_1,\dots,X_N \stackrel{\mathrm{iid}}{\sim} P$, $N \ge 2$, define the
$U$-statistic estimator
\begin{equation}
  \widehat{\mathrm{ES}}_N
  = \frac{1}{N}\sum_{i=1}^N \|X_i - y\|
  - \frac{1}{2N(N-1)} \sum_{i \ne j} \|X_i - X_j\| .
  \label{eq:es-hat}
\end{equation}
Then $\E\, \widehat{\mathrm{ES}}_N = \mathrm{ES}(P,y)$ for every $N \ge 2$.
\end{proposition}

In plain language, increasing $N$ makes the simulated energy score less noisy but
does not change its average target. This is the key contrast with MRMin.

\noindent\textit{Proof outline.} The first average is unbiased term by term. The second sum contains exactly $N(N-1)$ ordered independent pairs, so its $1/[2N(N-1)]$ normalisation has expectation $\tfrac12\E\|X-X'\|$. The full proof is in the Electronic Companion.

\begin{corollary}[Two distinct roles for $N$]
\label{cor:es-vs-mrmin}
For the energy score, Proposition~\ref{prop:es-unbiased} gives the same expected
objective for every integer $N\ge2$: ensemble size controls Monte Carlo precision
and does not redefine the estimand. In contrast, the MRMin limit is
\begin{equation}
\begin{aligned}
 \E\widehat{\mathrm{ES}}_N&=\mathrm{ES}(P,y)
 &&\text{for every }N\ge2,\\
 \lim_{N\to\infty}\underline\sigma_N
 &=\sigma_0\sqrt{1+2/K}>\sigma_0.
\end{aligned}
 \label{eq:es-vs-mrmin}
\end{equation}
MRMin is qualitatively different in two
respects.  Within achievable computational budgets, $\underline\sigma_N$ moves
with $N$ wherever the condition of Theorem~\ref{thm:mrmin} holds, as verified
over the trajectory-scale grids in Section~\ref{sec:exp1}.  In the ultimate
nearest-neighbour regime it does not converge to $\sigma_0$, but to
$\sigma_0\sqrt{1+2/K}$ in \eqref{eq:mrmin-finite-limit}, which exceeds
$\sigma_0$ for every finite $K$.  MRMin therefore targets an $N$-indexed
estimand at finite budget and a misspecified one in the limit; neither defect
is removed by increasing $N$.
\end{corollary}

\begin{remark}[Magnitude of the two defects]
The defects have opposite dimension dependence.  The limiting relative bias
$\sigma_\infty/\sigma_0-1$ is $73\%$ at $K=1$ but only $0.08\%$ at $K=1200$,
so it is negligible for trajectory calibration.  Finite-budget movement behaves
in the opposite way and dominates at large $K$ (Section~\ref{sec:exp1}).  For
spacing trajectories with $K$ of order $10^3$, it is therefore movement of the
estimand with $N$, not its asymptotic bias, that makes the target depend on a
computational choice.
\end{remark}

\begin{remark}[The normalisation matters]
The $U$-statistic normalisation $1/(N(N-1))$ in \eqref{eq:es-hat} is not a
cosmetic choice. Using $1/N^2$ in its place, as one might do naively by analogy
with a plug-in sample average, gives instead
\[
\begin{aligned}
 &\E\left[\frac{1}{2N^2}
       \sum_{i\ne j}\|X_i-X_j\|\right]\\
 &\quad=\frac{N(N-1)}{2N^2}\,\E\|X-X'\|\\
 &\quad=\frac{N-1}{2N}\,\E\|X-X'\|,
\end{aligned}
\]
so that
\[
\begin{aligned}
 \E\big[\widehat{\mathrm{ES}}_N^{\,(1/N^2)}\big]
  ={}&\E\|X-y\|\\
     &-\frac{N-1}{2N}\,\E\|X-X'\|\\
  ={}&\mathrm{ES}(P,y)
     +\frac{1}{2N}\,\E\|X-X'\|,
\end{aligned}
\]
which reintroduces an $O(1/N)$ bias and, with it, the $N$-dependence that the
energy score is meant to eliminate. The $U$-statistic form in
\eqref{eq:es-hat}, rather than the superficially similar plug-in average, is
therefore required.
\end{remark}

\begin{remark}[Unbiased objective versus parameter estimate]
Proposition~\ref{prop:es-unbiased} is an exact statement about the Monte Carlo
\emph{objective}: at each fixed $\beta$, its expectation is unchanged by $N$.
It does not say that the parameter value minimising one noisy simulated objective is
exactly unbiased. Because minimisation is nonlinear, the resulting estimate generally
retains a bias of order $1/N$ under standard smoothness conditions. Thus $N$
does not redefine the population minimiser, as it does for MRMin, but it can still
affect finite-ensemble parameter estimates. Section~\ref{sec:exp2} distinguishes
these two levels of $N$-dependence empirically.
\end{remark}

The energy score is selected from the wider class of proper scores for three
task-specific reasons. It can be evaluated from simulator draws without a
tractable predictive density, it scores the complete multivariate trajectory,
and its pairwise term has the exactly unbiased finite-ensemble form in
\eqref{eq:es-hat}. A log score would require a reliable trajectory density;
moment scores such as the Dawid--Sebastiani score do not strictly identify an
otherwise unrestricted non-Gaussian distribution; and univariate CRPS would
require an additional aggregation or projection choice across time. The
variogram score below is retained as a dependence-sensitive complement rather
than a replacement because it is not strictly proper for the full joint
distribution \citep{gneiting2007strictly,jordan2019scoringrules}.

\subsection{Correlation structure: the variogram score}
\label{sec:theory-vs}

The energy score treats the predictive distribution through pairwise Euclidean
distances and is, as a consequence, comparatively insensitive to the
correlation structure across the $K$ components of a trajectory. This matters
for spacing trajectories specifically: spacing residuals are the cumulative sum
of velocity residuals, so they are strongly autocorrelated by construction, and
the aleatoric errors observed in this class of models are empirically
non-Gaussian --- fat-tailed and left-skewed, as documented for these data by
\citet{zhou2025calibration}.

To make the correlation structure itself an object of the evaluation, we report,
alongside the energy score, the variogram score of order $p$
\citep{scheuerer2015variogram},
\[
\begin{aligned}
 \mathrm{VS}_p(P,y)
   &=\sum_{k,l}w_{kl}\,d_{kl}^{2},\\
 d_{kl}
   &=|y_k-y_l|^p-\E|X_k-X_l|^p,\\
 p&=0.5 .
\end{aligned}
\]
where the sum runs over pairs of time indices $k,l \in \{1,\dots,K\}$ and
$w_{kl} \ge 0$ are fixed weights.

The diagnostic value of reporting both scores is direct: agreement between the
parameter estimates obtained by minimising $\mathrm{ES}$ and by minimising
$\mathrm{VS}_p$ indicates that the model's correlation structure is adequate for
the aspects of the data the two scores are each sensitive to. Disagreement is
direct evidence against the assumed noise structure --- for QIDM, against the
assumption that the underlying acceleration noise is white, since a violation of
that assumption changes the autocorrelation of spacing without necessarily
changing its marginal dispersion, a distinction $\mathrm{ES}$ is not designed to
detect and $\mathrm{VS}_p$ is. Section~\ref{sec:exp7} uses this disagreement as a
mechanism diagnostic.

\subsection{How uncertainty decreases with repeated trajectories}
\label{sec:theory-asymptotics}

Write $s(\beta,y)=\mathrm{ES}(P_\beta,y)$ and
$\bar S_n(\beta)=n^{-1}\sum_{j=1}^n s(\beta,y_j)$ for the exact
sample-averaged energy score over $n$ independent observed trajectories, and
\[
  \hat\beta_n = \argmin_\beta \bar S_n(\beta) .
\]
The following result describes how parameter uncertainty decreases as more independent
observed trajectories are added. This estimator belongs to the standard class of
M-estimators. Under the usual regularity conditions
\citep{vandervaart1998asymptotic} --- in particular, that $\beta_0$ is an
interior point of the parameter space, that $\bar S_n$ is twice continuously
differentiable in a neighbourhood of $\beta_0$, and that the usual
uniform-convergence and identifiability conditions hold so that $\hat\beta_n$ is
consistent --- the standard M-estimation sandwich expansion gives
\begin{equation}
\begin{aligned}
  \sqrt{n}\,(\hat\beta_n - \beta_0)
  &\;\rightsquigarrow\;
  \mathcal{N}\big(0,\; A^{-1} B A^{-1}\big),\\
  A&=\E\{\nabla^2_\beta s(\beta_0,Y)\},\\
  B&=\Var\{\nabla_\beta s(\beta_0,Y)\}.
\end{aligned}
  \label{eq:sandwich}
\end{equation}
Here $A$ measures how sharply the expected score curves around its minimum, while
$B$ measures how much its slope varies from one observed trajectory to another. The
combination $V_1=A^{-1}BA^{-1}$ is called the sandwich covariance. It describes the
large-sample uncertainty scale for one trajectory. Averaging $n$ independent
trajectories reduces the covariance to approximately $V_1/n$, so the standard errors
fall at rate $1/\sqrt n$.

One technical point deserves explicit statement rather than proof:
$\|\cdot\|$ is non-differentiable at the origin, so $\nabla_\beta
\mathrm{ES}(P_\beta,y)$ requires differentiating $\E\|X-y\|$ and $\E\|X-X'\|$
under the expectation, which in turn requires $P_\beta$ to admit a density (so
that $X = y$ and $X = X'$ are almost-sure non-events and the non-smooth point of
the norm is not visited with positive probability under the relevant measure).
We state this as a condition on $P_\beta$ rather than as a technical lemma to be
proved, since every stochastic CF model considered here generates trajectories
via absolutely continuous noise and satisfies it by construction.

Proposition~\ref{prop:es-unbiased} concerns the expectation of the finite-ensemble
objective at a fixed $\beta$; it does not remove simulation noise from its gradient
or Hessian. In computation, replacing $s$ by $\widehat s_N$ with fixed $N$ adds a
Monte Carlo component whose contribution depends on whether simulations are refreshed
or common random numbers are held fixed. Equation~\eqref{eq:sandwich} therefore
describes the exact-score limit. A finite-$N$ numerical sandwich either requires an
additional simulation-variance term or should be interpreted conditionally on the
fixed random-number streams; it converges to the exact-score expression as the
simulation approximation and its derivatives converge uniformly.

The present dataset supplies up to 12 repeated trajectories per driver. Their
per-run parameter estimates provide an empirical \emph{single-run} covariance
benchmark, not the covariance of the $n$-run average estimator in
\eqref{eq:sandwich}. Moreover, the available published per-run estimates use the
earlier one-step estimation procedure rather than the same energy-score estimator.
Section~\ref{sec:exp8} therefore treats their covariance as an external stability
benchmark for the local, fixed-common-random-number energy-score curvature, not as a
second consistent estimator of the identical covariance matrix.

\subsubsection{On resampling from a single trajectory}
\label{sec:theory-resampling}

\citet{zhou2025calibration} estimate the parameter covariance matrix by drawing
$B = 100$ overlapping blocks of length $L = 500$ from a single trajectory of
length $T = 1200$, with block start points drawn uniformly on $[1, 701]$ (so
that a block of length $500$ starting anywhere in that range stays within the
series). Two such blocks overlap unless their start points differ by at least
$L=500$.  For independent continuous uniforms on an interval of length $R$,
\[
  \Pr(|U_1-U_2|\geq L)=\left(1-\frac{L}{R}\right)^2,
  \qquad L\leq R.
\]
Here $R=700$ and $L=500$, hence the continuous convention gives
$(1-5/7)^2=4/49\approx0.0816$.  A direct reconstruction of the integer-start
scheme in \citet{zhou2025calibration}, using 20,000 sets of $B=100$ starts from
$\{1,\ldots,701\}$, gives empirical probability $0.082694$; the exact discrete
probability is $0.082625$.  The small difference is solely the endpoint
convention. Under either convention non-overlap occurs only about $0.082$ of
the time and the great majority of block pairs share data. Overlap is legitimate
within a moving-block bootstrap, but the resulting block-specific estimates cannot
be interpreted as independent parameter repetitions without accounting for their
dependence. Separately, the chosen block length $L = 500$ is $0.42\,T$. Orders such
as $T^{1/3}$ arise for particular variance and bias problems
\citep{hall1995blocking}; they are not universal, but they illustrate that an
asymptotically valid block length normally shrinks relative to $T$ and must be
matched to the statistic and dependence structure.

The repeated-measurement structure of the present dataset --- up to 12 independent
runs per driver under an identical stimulus --- provides a direct empirical
single-run stability benchmark without choosing a block length. As stated in
Section~\ref{sec:theory-asymptotics}, the available published per-run estimates do
not use the same estimator as the energy-score sandwich calculation, so the two are
reported as complementary diagnostics rather than interchangeable covariance
estimates.
The known-truth coverage experiment in Section~\ref{sec:exp8} quantifies the
difference. At nominal levels $0.50$, $0.80$, $0.90$, and $0.95$, the
overlapping-block intervals cover $0.447$, $0.681$, $0.763$, and $0.800$,
respectively, whereas intervals based on 12 independent repetitions cover
$0.475$, $0.797$, $0.895$, and $0.956$; the complete curves appear in
Figure~\ref{fig:exp8-coverage}.

\section{Synthetic validation of the propriety results}
\label{sec:synthetic}

\subsection{Experimental design}
\label{sec:synthetic-design}

The synthetic design directly stress-tests the fixed-$N$, discretised-grid evidence
that supported MRMin in \citet{zhou2025calibration}. It makes recovery harder in ways
that reveal the inferential properties of the criteria rather than their ability to
exploit a convenient grid.
Ground-truth parameters are placed off the evaluation grid, preventing exact recovery
by construction. The main joint-recovery calibrations (Section~\ref{sec:exp3})
use a bounded continuous optimiser rather than an
exhaustive search, so the result reflects the objective surface and not the spacing of
a candidate set; the two-parameter joint diagnostic of Section~\ref{sec:exp2} is the
one exception, and is deliberately grid-based for a reason given there. All
behavioural parameters, including desired speed $v_0$, are
estimated jointly with the stochastic parameters. Finally, measured-magnitude GNSS
error is added to the synthetic observations in the recovery experiments, separating
process variation from observation noise under the conditions faced by the empirical
analysis. These choices strengthen the test relative to designs in which the truth is
a grid point or deterministic parameters are held fixed.

The lead trajectories are selected before inspecting calibration outcomes and span
free acceleration, cruising, following, deceleration, and near-standstill regimes.
Because different parameters become informative in different regimes, recovery is
assessed over the complete trajectory rather than a hand-selected interval. Common
random numbers are used within an optimisation to reduce simulation noise while
independent replicate observations measure sampling variation. Table~\ref{tab:models}
fixes the model scope and parameter domains for the joint-recovery experiment.

\begin{table*}[t]
\centering
\caption{Stochastic car-following models, calibrated parameters, ground truths, and
bounds used in the joint-recovery study. Deterministic and stochastic parameters are estimated
jointly; speeds are in km/h and all remaining quantities use SI-compatible units.}
\label{tab:models}
\begin{tabular}{@{}p{0.18\linewidth}p{0.26\linewidth}p{0.19\linewidth}p{0.25\linewidth}@{}}
\toprule
Model & Calibrated parameters & Ground truth & Lower and upper bounds \\
\midrule
QIDM & $v_0,a,b,s_0,T,Q$ & 73.1, 1.37, 2.63, 1.87, 0.77, 0.47 & [40,100], [0.5,3], [0.5,5], [0.5,5], [0.1,1], [0.02,2] \\
2D-IDM & $v_0,a,b,s_0,T_{\min},\Delta T$ & 73.1, 1.37, 2.63, 1.87, 0.53, 0.47 & [40,100], [0.5,3], [0.5,5], [0.5,5], [0.1,1], [0.01,1] \\
\bottomrule
\end{tabular}
\end{table*}

\subsection{Numerical verification of the Gaussian theory}
\label{sec:exp1}

The Gaussian experiment first separates the score-level mechanisms that appeared
together in the model-specific error analysis of \citet{zhou2025calibration}. It
supplies numerical shadows of Section~\ref{sec:theory},
not independent evidence for its analytic statements.  We set $\sigma_0=1$ and
evaluate the population formula whenever it is available.  For MRMin, scalar
Gaussian-tail quadrature is used at $K=1$; for $2\leq K\leq20$,
generalised Gauss--Laguerre radial quadrature is combined with conditional
noncentral-chi-square quadrature, while fixed scrambled-Sobol radial draws are
retained at larger $K$.
Common draws across $(N,\sigma)$ prevent optimiser movement from being driven by
independent Monte Carlo noise.

Figure~\ref{fig:exp1-scores} separates five views of the objective geometry.
MRMean-I is minimised at zero scale, MRMean-II is flat, and the energy score is
uniquely minimised at the data-generating scale.  MRMin instead changes target
with $N$.  Panel (e) isolates the practical mean run-wise RMSE from
Remark~\ref{rem:gof-generalises}: it and the proposed criterion differ by the
single pairwise term $-\tfrac12\E\|X-X'\|$, and that term determines whether the
population objective collapses at zero dispersion or has an interior minimum
at $\sigma_0$.

\begin{figure}[t]
\centering
\includegraphics[width=0.74\linewidth]{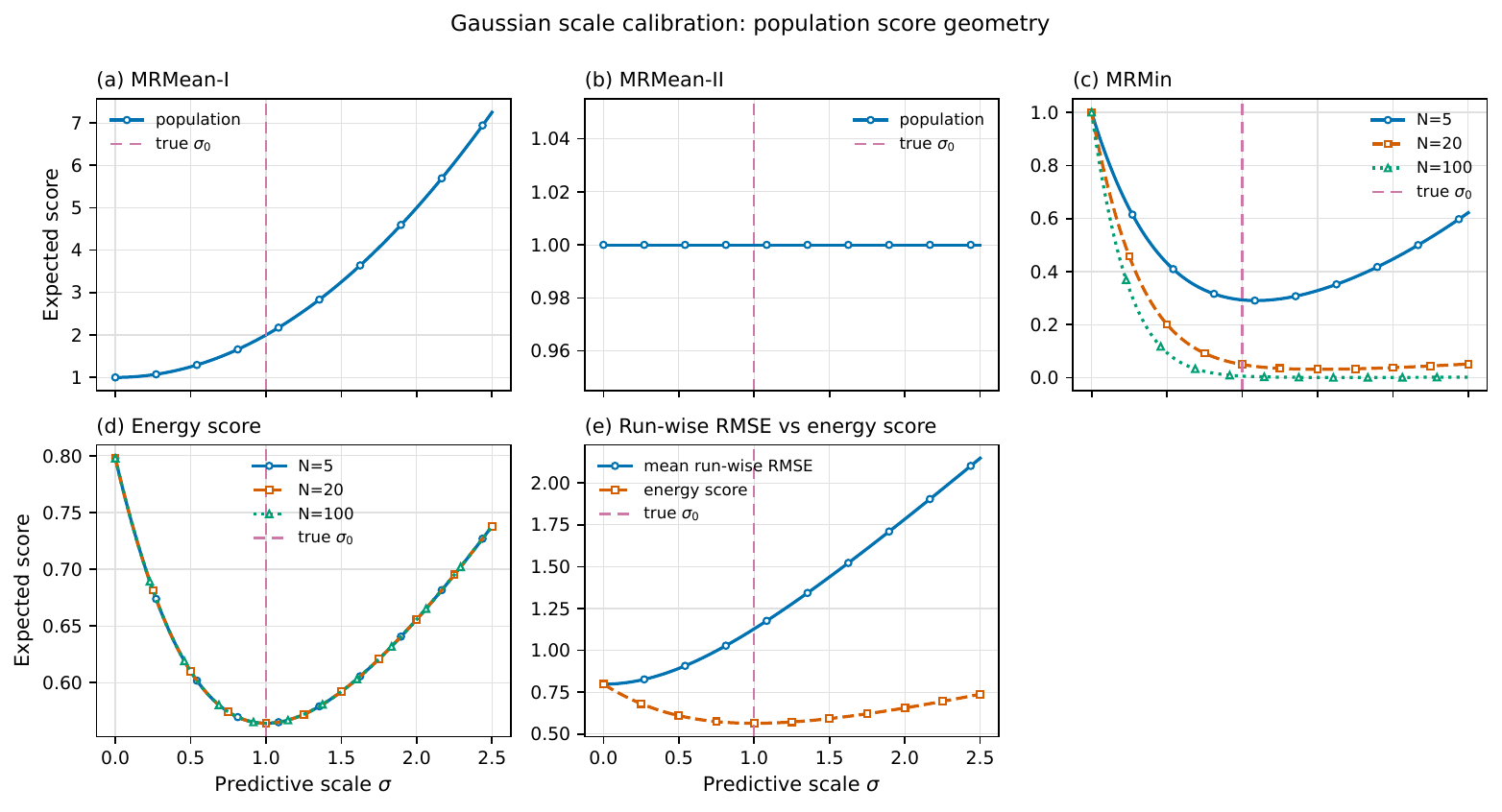}
\caption{Expected score against predictive scale $\sigma$ in the correctly
centred scalar Gaussian toy. Panels (a)--(d) show MRMean-I, MRMean-II, MRMin,
and the energy score; panel (e) compares mean run-wise RMSE with the energy
score. The vertical line marks $\sigma_0=1$. Coincident energy-score curves
show that ensemble size changes simulation precision but not the expected
energy-score objective.}
\label{fig:exp1-scores}
\end{figure}

Table~\ref{tab:theory-verification} collects the numerical checks.  The entries
test implementations and finite-range implications of the analytic results;
they do not replace the assumptions or proofs in Section~\ref{sec:theory}.

For the high-dimensional finite-budget approximation, 35 logarithmically spaced
values of $N=10$--$2000$ were fitted with three forms on the same response scale:
$cN^a$, $a+b\log N$, and $a+b\sqrt{\log N}$. The square-root form has the largest
$R^2$ at every tested $K\in\{1,10,100,1200\}$ (Table~\ref{tab:theory-verification}),
and the finite-range growth exponents it implies fall far below the historical
$N^{1/K}$ reference except near $K=10$: at $K=1200$, the reference predicts
$2000^{1/1200}-1=0.6\%$ growth over the tested range against an observed movement
exceeding $130\%$. Figure~\ref{fig:exp1-dimension} plots the minimisers, fitted
curves, and reference slopes for the full comparison.

At twelve fixed $(K,\sigma_1)$ combinations, $g_N(\sigma_1)$ decreases strictly
over $N=10,\ldots,2000$, as expected from nearest-neighbour coverage.  This
pointwise decrease is not a divergence diagnostic: the small-scale limit is not
uniform in $N$, and the scalar extension instead approaches the finite value in
\eqref{eq:mrmin-finite-limit}.

The scalar large-$N$ calculation gives the particularly sharp check
$\sqrt{1+2/1}=\sqrt3=1.73205$: numerical minimisers are already
$1.731$--$1.732$ at $N=10^4$--$10^5$.  This agreement is the strongest direct
numerical confirmation of the nearest-neighbour derivation
\eqref{eq:mrmin-nn-asymptotic}--\eqref{eq:mrmin-finite-limit}.  The extended
decreasing-differences grid is more qualified.  Every one of the 12
configurations with $K\geq10$ has a negative adjacent-difference fraction of
$1.000$, while scalar fractions range from $0.378$ to $0.815$; hence that route
is a sufficient conditional argument, not a global scalar proof.

The final two rows of Table~\ref{tab:theory-verification} impose a direct
constraint on implementation: replacing $1/[N(N-1)]$ by $1/N^2$ introduces a
bias exactly proportional to $1/N$, reinstating the very $N$-dependence the
method is designed to remove.

\subsection{Ensemble-size sensitivity in car-following calibration}
\label{sec:exp2}

The earlier MRMin comparison fixed $N$ and therefore could not reveal whether the
method recovered a stable parameter or an $N$-specific target. This experiment makes
$N$ the controlled variable and translates the three objective pathologies into
calibration signatures.  For a pure dispersion coordinate, MRMean-I predicts a
boundary estimate, MRMean-II predicts weak identification governed by finite
simulation noise and bounds, and MRMin predicts motion with $N$.  The energy
score has an $N$-invariant population objective, although its finite-sample
argmin can still vary.

For QIDM, the stochastic intensity is a pure scale coordinate under the
controlled design. MRMean-I reaches the zero boundary, MRMean-II remains broadly
dispersed, and MRMin changes with $N$. The energy-score median at $N=500$ is
$0.382$ for one observed trajectory, against $Q_0=0.47$. Increasing the number
of observed trajectories to $2,4,8,$ and $12$ gives medians $0.453$, $0.467$,
$0.432$, and $0.453$. Most of the $n=1$ bias disappears, without monotonic
finite-sample convergence; this is argmin sampling error rather than a moving
population minimiser.

Example~1 in Section~\ref{sec:mrmean2} predicts a different geometry for
2D-IDM.  With $T_{\min}$ held at its true value, every criterion can return
$\Delta T_0$ through the mean channel, so the design cannot discriminate among
dispersion criteria.  The observed agreement near
$\widehat{\Delta T}=0.727$ against $\Delta T_0=0.73$ is consistent with $b=0$.

The prediction was tested by deliberately fixing the fitted lower endpoint at
$T_{\min}=0.7$ rather than its truth $0.6$, so $b=0.1$.  Example~1 predicts
$\widehat{\Delta T}=0.53$.  Over 50 paired MRMean-II calibrations with $N=1000$,
the estimate had median $0.54583$, mean $0.54601$, and IQR $0.01762$: the shift
is close to the explicit mean-channel prediction, with the residual difference
attributable to nonlinear trajectory dynamics and a single verifying path per
replicate.

The same observations were then calibrated jointly in $(T_{\min},\Delta T)$ under
all four criteria. At $N=2000$, coordinate correlations were $-0.9949$, $-0.9971$,
$-1.0000$, and $-0.9969$, with median pairs $(0.82,0.307)$, $(0.73,0.505)$,
$(0.82,0.307)$, and $(0.55,0.802)$, for MRMean-I, MRMean-II, MRMin, and the energy
score respectively. Because this diagnostic deliberately evaluates the full objective
surface on an $11\times11$ grid with local quadratic refinement, rather than the
continuous optimiser of Section~\ref{sec:exp3}, MRMean-I and MRMin return an
identical median pair at every tested $N$: a grid-resolution artefact, since both
objectives are locally flat enough near their optimum to fall back to the same
vertex, not a substantive equivalence. The coarse-grid ranking by distance to
$(0.60,0.73)$ is accordingly not treated as a recovery result.

To resolve the discretisation issue, a second diagnostic evaluates 51 values at
0.01-s spacing along the known data-generating line
$T_{\min}+\Delta T/2=0.965$, which includes the truth exactly, for 50 new
$N=2000$ calibrations. The resulting median pairs are $(0.789,0.352)$,
$(0.616,0.698)$, $(0.678,0.574)$, and $(0.623,0.685)$ for MRMean-I,
MRMean-II, MRMin, and the energy score, with Euclidean errors $0.422$, $0.036$,
$0.174$, and $0.050$~s. MRMean-II's close median is not identification: the
diagnostic supplies the true mean constraint that its population objective can
otherwise learn without determining width. The energy score has the narrowest
profile distribution (IQR $0.073$~s in $T_{\min}$, equivalently $0.146$~s in
$\Delta T$), compared with $0.088$ and $0.176$~s for MRMean-II. Thus the robust
finding is sharper conditional concentration from the dispersion reward, not
that one coarse coordinate median alone proves recovery. Individual estimates
remain spread along the weakly identified ridge.

Equivalently, define the mean and width coordinates
\[
  m_T=T_{\min}+\frac{\Delta T}{2},
  \qquad w_T=\Delta T.
\]
The experiment identifies $m_T$ much more strongly than $w_T$: the diagonal in
Figure~\ref{fig:exp2-joint} is a constant-$m_T$ line, and movement along it changes
the stochastic width while preserving the dominant mean trajectory. We therefore
interpret the joint ridge and the profile in $(m_T,w_T)$ terms; fixing $T_{\min}$
at truth would artificially remove the weakly identified direction rather than
demonstrate that a criterion recovered it.

The energy-score estimator uses the unordered-pair sum divided by $N(N-1)$; as an
implementation check, the MRMin optimum increases monotonically with $N$ in the
scalar Gaussian setting of Theorem~\ref{thm:mrmin} ($0.55$, $0.95$, $1.275$, $1.325$
over $N=2,5,20,100$), and repeated evaluations of the QIDM energy-score objective
under identical common random numbers reproduce the objective value bit-for-bit.
Figure~\ref{fig:exp2} shows the marginal $N$-sensitivity signatures,
Figure~\ref{fig:exp2-joint} shows the joint identification ridge, and
Table~\ref{tab:exp2} reports the corresponding distributional summaries.

\begin{figure}[t]
\centering
\includegraphics[width=0.74\linewidth]{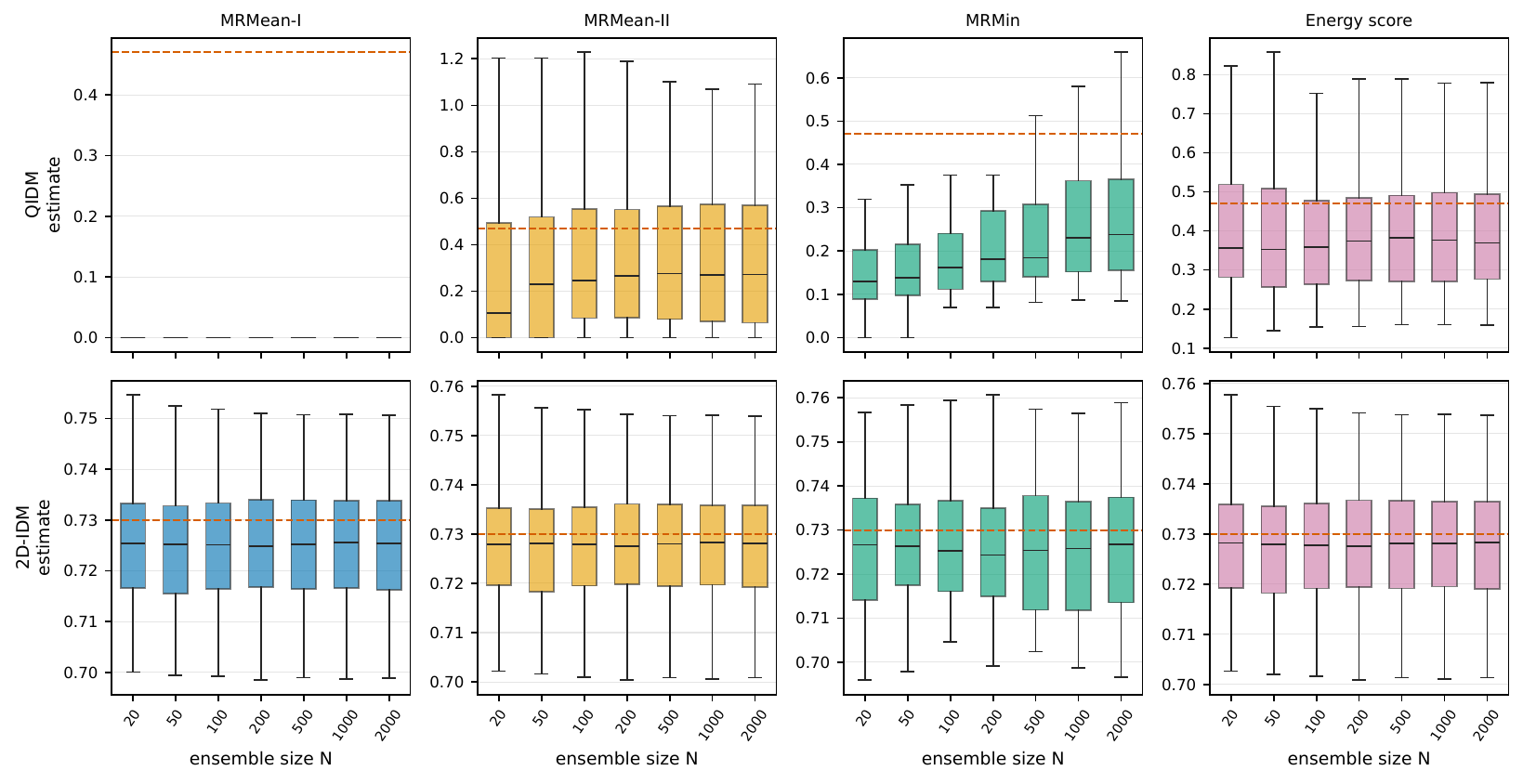}
\caption{Estimated stochastic parameter against ensemble size $N$ for QIDM and
2D-IDM over 50 paired calibrations.  The horizontal line marks the data-generating
value.  The 2D-IDM agreement is interpreted through the mean-channel diagnostic
of Example~1 rather than as evidence that all criteria identify dispersion.}
\label{fig:exp2}
\end{figure}

\begin{figure}[t]
\centering
\includegraphics[width=0.72\linewidth]{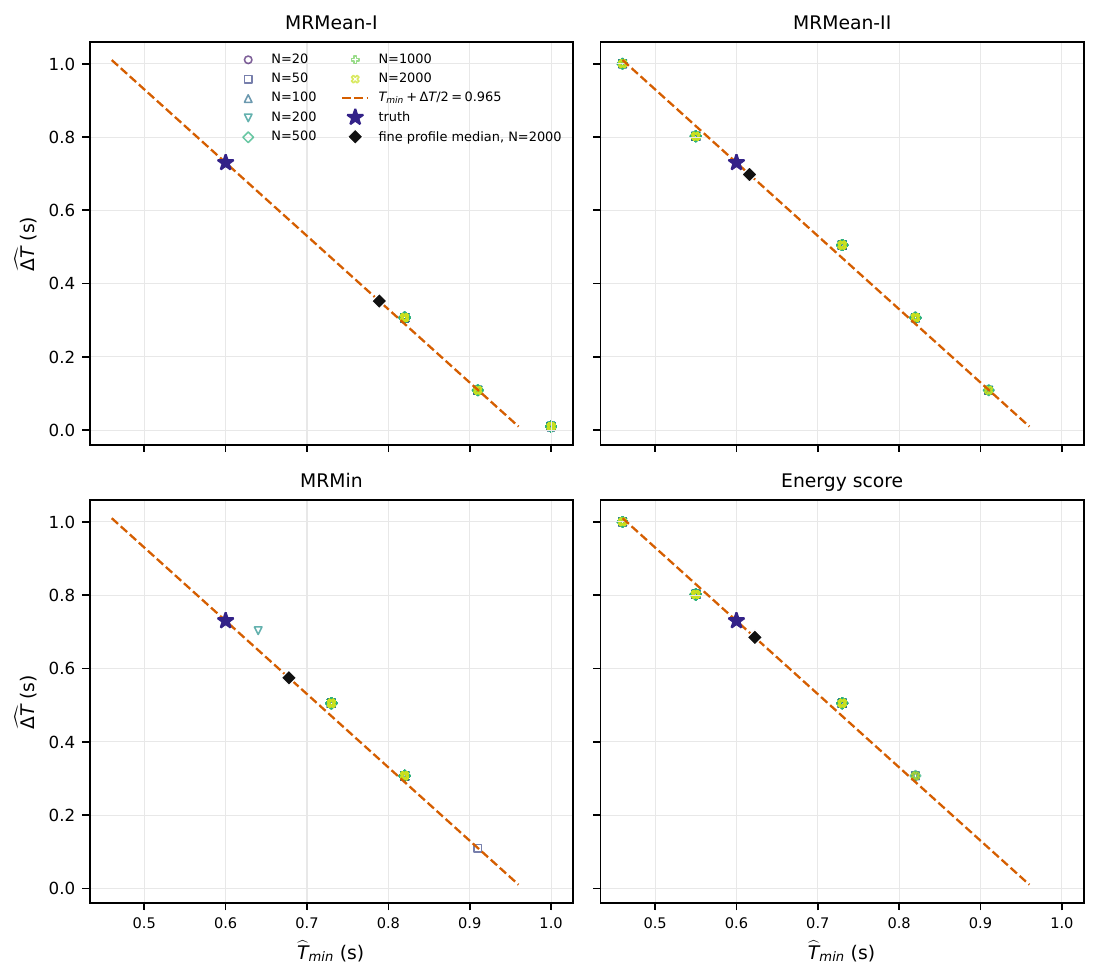}
\caption{Joint 2D-IDM estimates under the four criteria. Open symbols show 50
repetitions at each $N$ on the original $11\times11$ two-dimensional grid. The
star marks the off-grid truth $(T_{\min},\Delta T)=(0.60,0.73)$; the diagonal is
the iso-mean line $T_{\min}+\Delta T/2=0.965$; and the black diamond marks the
median from the separate 51-point, $N=2000$ fine profile conditional on that
line. Marginal $\Delta T$ values and coarse-grid distance rankings are not
interpreted apart from this joint geometry.}
\label{fig:exp2-joint}
\end{figure}

\subsection{Joint recovery of deterministic and stochastic parameters}
\label{sec:exp3}

The near-zero MRMin errors reported by \citet{zhou2025calibration} were obtained by
traversal over a discretised parameter space containing the planted values. This
experiment asks whether that ranking survives when exact grid recovery is unavailable.
It tests whether the one-parameter signatures survive joint
calibration of deterministic and stochastic coordinates under observation
noise. Each model has the fixed off-grid truth reported in
Table~\ref{tab:models}; independent observation paths are generated for every
replicate and criterion. Recovery is summarised by parameter-wise absolute percentage
error (APE). Because an overall AAPE gives five deterministic coordinates five times
the weight of the single stochastic coordinate, Table~\ref{tab:exp3} reports
deterministic-coordinate AAPE and stochastic-coordinate APE separately. It also
reports range-normalised mean absolute error,
\[
  100p^{-1}\sum_{j=1}^{p}
  \frac{|\widehat\theta_j-\theta_{0j}|}{U_j-L_j},
\]
where $[L_j,U_j]$ is the prespecified calibration range. This third measure is
dimensionless without dividing by a potentially small planted value. Overall AAPE is
retained in the numerical archive for comparison with the earlier study but is not
used as the headline ranking. Deterministic
parameters are estimated jointly with the stochastic coordinates, rather than
being held at their data-generating values, so a low total error cannot conceal
failure concentrated in dispersion.  Common random numbers are fixed within an
optimisation and independent observations measure sampling variation. The design
contains 160 calibrations: two models, four criteria, two tests, and ten
repetitions, using two lead profiles. The main GNSS spacing-error SD is $0.20$ m,
with sensitivity runs at $0.10$ and $0.35$ m. Energy score is not uniformly best in
joint finite-sample recovery. It gives the smallest QIDM stochastic-coordinate error,
whereas 2D-IDM is best under MRMean-II because $\Delta T$ also moves the mean channel.
Range normalisation does not reverse those rankings.
Figure~\ref{fig:exp3} resolves these errors by parameter and model, while
Table~\ref{tab:exp3} separates the deterministic, stochastic, and joint summaries.

\subsection{When exact MLE is available}
\label{sec:exp4}

MLE can be exact or approximate, and the distinction matters. Exact MLE requires the
model to give the probability density of the observed data. For a time-stepping
model, this usually means that we can evaluate the probability of the next observed
state given the current state and the parameters. \citet{zhou2025calibration} instead
used a Gaussian approximation for complete-trajectory residuals and treated residuals
at different times as independent with constant variance. That approximation is not
the probability law used by QIDM or 2D-IDM to generate the next state. Its poor
recovery therefore does not show that MLE removes stochasticity; it shows that an
incorrect likelihood can give an incorrect estimate.

The fourth experiment creates a controlled case in which this next-state density is
known and easy to evaluate. QIDM is fitted in two ways: by exact one-step MLE and by
the four criteria applied to complete simulated trajectories. If exact MLE recovers
the parameter but a trajectory criterion does not, the stochastic parameter is not
inherently impossible to estimate; the difference comes from the calibration
objective or from information lost when one complete trajectory is reduced to a
single score.

We call this an \emph{oracle conditional} experiment because it deliberately gives
MLE information that would not all be known in a normal application: the deterministic
IDM parameters are fixed at their true values, the current state has no added GNSS
error, and acceleration is reconstructed with the same update equation that generated
the synthetic path. Given this information, the one-step errors satisfy
$e_k\sim\mathcal N(0,Q/\Delta t)$ exactly, giving the closed-form estimator
\[
  \widehat Q_{\mathrm{MLE}}=\Delta t\,K^{-1}\sum_{k=1}^{K}e_k^2.
\]
If the state is measured with error or the deterministic parameters are also unknown,
this simple formula is no longer the exact likelihood; a state-space model or another
method for latent states would be needed.

The comparison uses the same off-grid truth, observation paths, parameter bounds,
and reporting metrics across methods. Exact one-step MLE in the controlled setting recovers $Q$
essentially without bias, whereas full-trajectory criteria reflect their
distinct targets and the information loss from one realised path rather than
a normalisation error or an unstable finite-ensemble argmin. The result
establishes a practical distinction: use exact MLE when the model can evaluate
the probability of each observed state change; use a proper trajectory score when
the simulator can generate trajectories but cannot evaluate their probability.
Figure~\ref{fig:exp4} displays the full sampling distributions of $\widehat Q$;
Table~\ref{tab:exp4} reports their median, bias, and IQR.

\begin{figure}[t]
\centering
\includegraphics[width=0.68\linewidth]{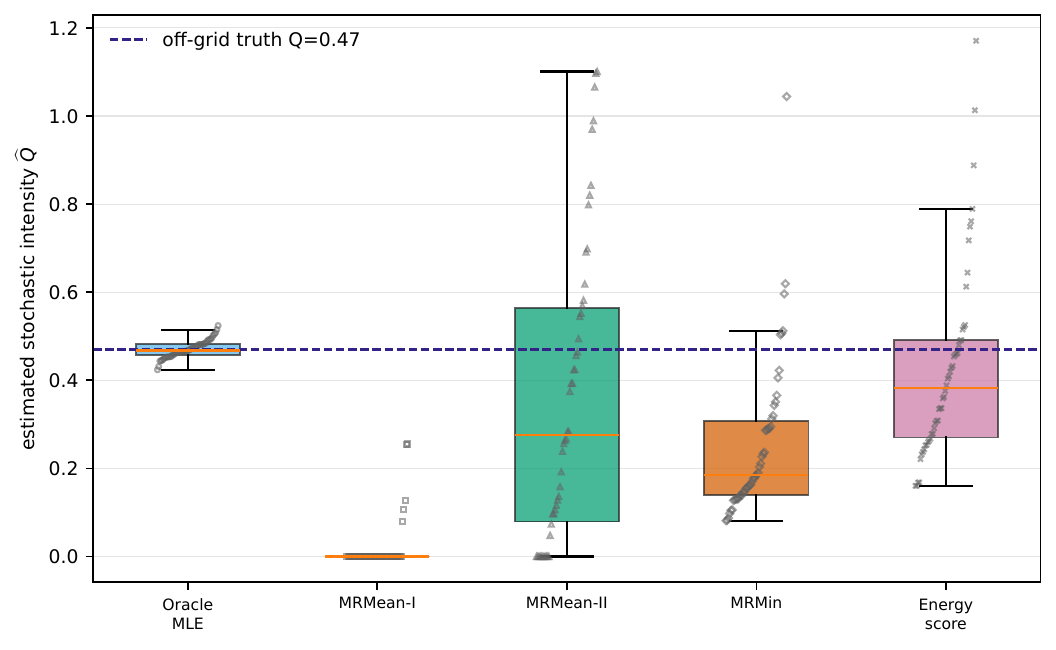}
\caption{Recovery of the QIDM stochastic intensity $Q$ under exact one-step MLE in
the controlled setting and multi-step simulation-based calibration. Points show
independent estimates, boxes show their distributions, and the horizontal line
marks the off-grid truth $Q_0=0.47$. Each method has 50 independent estimates;
simulation criteria use $N=500$. The controlled case fixes deterministic parameters at truth
and uses synthetic states without added measurement error.}
\label{fig:exp4}
\end{figure}

\begin{table}[t]
\centering
\caption{QIDM recovery under the controlled exact one-step MLE baseline and simulation-based
scoring rules.}
\label{tab:exp4}
\begin{tabular}{lrrr}
\toprule
Criterion & Median $\widehat Q$ & Bias & IQR \\
\midrule
Exact one-step MLE (controlled setting) & 0.468 & -0.002 & 0.025 \\
MRMean-I & 0.000 & -0.470 & 0.000 \\
MRMean-II & 0.275 & -0.195 & 0.486 \\
MRMin & 0.185 & -0.285 & 0.168 \\
Energy score & 0.382 & -0.088 & 0.221 \\
\bottomrule

\end{tabular}
\end{table}

\subsection{Computational cost}
\label{sec:cost}

Strict propriety is useful only if its computational cost is explicit. All experiments
were run with Python~3.12.13, NumPy~2.4.6, and SciPy~1.18.0 on an AMD Ryzen~7 H~260
(8 cores, 16 logical processors), holding common-random-number policies constant.
MRMean requires $O(NK)$ trajectory comparisons, exact MRMin has the same leading
comparison order after simulation, and a direct energy-score implementation adds an
$O(N^2K)$ pairwise term. A dedicated single-process score-only benchmark at
$K=1200$ excludes simulation and optimisation. At $N=20,100,500,$ and $2000$,
median energy-score evaluation times are $0.129$, $2.222$, $48.866$, and
$823.573$ ms, versus $0.027$, $0.385$, $2.324$, and $10.303$ ms for MRMin;
the full method-by-$N$ timings are archived in the numerical audit. Whole-calibration
wall times are not compared because the objective landscapes trigger different
numbers of optimiser evaluations. The $O(N^2K)$ pairwise term is the practical cost of
propriety, and production use of the energy score should chunk the exact
pairwise calculation or use a separately validated unbiased pair estimator.

\section{Empirical validation with repeated-measurement trajectories}
\label{sec:experimental}

\subsection{Data and experimental setup}
\label{sec:experimental-setup}

The empirical study uses 11 human drivers following an autonomous lead vehicle over a
nominal 12-run repeated protocol.
The lead vehicle executed the same prescribed speed programme in each run, covering
acceleration, cruise, car-following, braking, and standstill. Position and speed were
recorded by GNSS at 10 Hz, with provider-stated speed accuracy $0.01$--$0.02$ km/h
and spacing accuracy $0.007$--$0.35$ m. The released signals were Fourier low-pass
filtered; filter-order metadata were unavailable and are not inferred.
Run-specific measurement-error variances were also unavailable, so they cannot be
entered as regression weights or covariates. The synthetic sensitivity at spacing-error
SD $0.10$, $0.20$, and $0.35$ m brackets the upper part of the stated range.
Front-to-front
spacing is retained as the measure of
performance, while the dynamical models subtract a vehicle length of
4.6 m when a net gap is required.

The design is repeated, but the released source is not a complete rectangular panel.
It contains 109 observed driver--run cells out of
132 nominal cells; absent cells were never recorded and are not post-hoc
exclusions. The retained lead trajectories follow the prescribed programme but are
not pointwise identical, so each simulation uses the measured lead path from its own
run. No present run is silently excluded. Driver~5/run~1 is the strongest spacing-DTW
candidate and driver~8/run~9 has the largest terminal-speed anomaly; both remain in the
main analysis, with filter and driver sensitivity reported separately.

Repeated measurement is the empirical distinction on which Sections~\ref{sec:exp5}--
\ref{sec:exp8} rely. Within a driver, observed repetitions form an ensemble that can
be compared with a simulated ensemble without designating one stochastic path as
truth. Independent runs also permit direct sample covariance of parameter estimates,
avoiding assumptions about overlapping blocks from one time series. Finally, the
common stimulus lets the analysis compare how variability changes across driving
regimes, which is what separates candidate stochastic mechanisms rather than merely
ranking their trajectory fit. This is a trajectory-level counterpart to the two-level
principle of \citet{punzo2020twolevel}: calibration of microscopic behaviour and
validation of the distribution it induces are distinct operations.

\subsection{Ensemble-to-ensemble distributional calibration}
\label{sec:exp5}

The earlier empirical analysis argued that a single MRMin-calibrated parameter set
could reproduce the range of repeated trajectories. The first empirical experiment
tests that claim out of sample and replaces single-trajectory comparison with a two-sample
energy distance between simulated and observed ensembles. For every driver and model,
the last two available repetitions are fixed as the test set. The stochastic
coordinate is then recalibrated either against one prespecified training trajectory
with ES or against the full training ensemble with two-sample energy distance;
deterministic coordinates are fixed at the training-run mean of independent one-step
estimates. Repeated-set calibration lowers held-out distance by a median $1.368$ for
QIDM ($p=0.041$) and $5.033$ for 2D-IDM ($p=0.226$). The distributional gain is
therefore resolved for QIDM but not at driver level for 2D-IDM.
Figure~\ref{fig:exp5} shows the paired driver-level held-out distances underlying
these median differences and permutation tests.

\subsection{Out-of-sample predictive performance}
\label{sec:exp6}

This experiment advances from validation-error summaries and selected prediction
profiles in \citet{zhou2025calibration} to held-out distributional diagnostics.
Coverage, PIT, spread ratio, and the variogram score are external to the four
calibration objectives; held-out energy score evaluates the energy-score criterion
on new data rather than introducing a different scoring functional.
A best-of-$N$ validation statistic would favour parameters that inflate dispersion
because a wider ensemble is more likely to contain an unusually close trajectory.
We instead evaluate coverage and width of predictive bands, PIT uniformity, the ratio
of simulated to observed standard deviation, and held-out energy and variogram
scores. The criterion-external diagnostics determine whether the ranking persists
beyond the score optimised during calibration.

The predictive-band panels compare nominal and empirical coverage across time
and driving regimes. A calibrated distribution should cover held-out paths without
achieving that coverage through excessive width; PIT histograms should be uniform
rather than U-shaped, peaked, or skewed. Table~\ref{tab:exp6} summarises all 109
leave-one-run-out folds from 11 drivers. The energy-score criterion gives the best
coverage ($42.1\%$ against the nominal $90\%$), PIT--KS ($0.505$), spread ratio
($0.738$), held-out energy score ($60.11$), and held-out variogram score ($610.84$),
followed by MRMin. This is a real result, not a
flattering one: even the best-performing criterion covers under half its nominal
target, and MRMean-I and MRMean-II are far worse still, at $0.001$ and $0.017$.
The gap is diagnostic rather than merely disappointing. A single scalar dispersion
parameter, calibrated against a stimulus with its own deterministic mismatch to the
model, cannot inflate the predictive band enough to cover that mismatch without
destroying the sharpness a proper score is designed to preserve; widening the band
further would trade one failure for another. The ranking among criteria is
informative evidence for which objective is least wrong, not a claim that any of
them is adequate for deployment as calibrated here. Calibrating the stochastic
intensity alone cannot repair deterministic structural error in the empirical
trajectories.
Figure~\ref{fig:exp6-bands} illustrates the time-resolved predictive bands for a
predefined typical fold, selected by proximity to the median energy-score coverage,
and Figure~\ref{fig:exp6-pit} shows
the pooled PIT shape used alongside the aggregate diagnostics in Table~\ref{tab:exp6}.

\begin{figure}[t]
\centering
\includegraphics[width=0.76\linewidth]{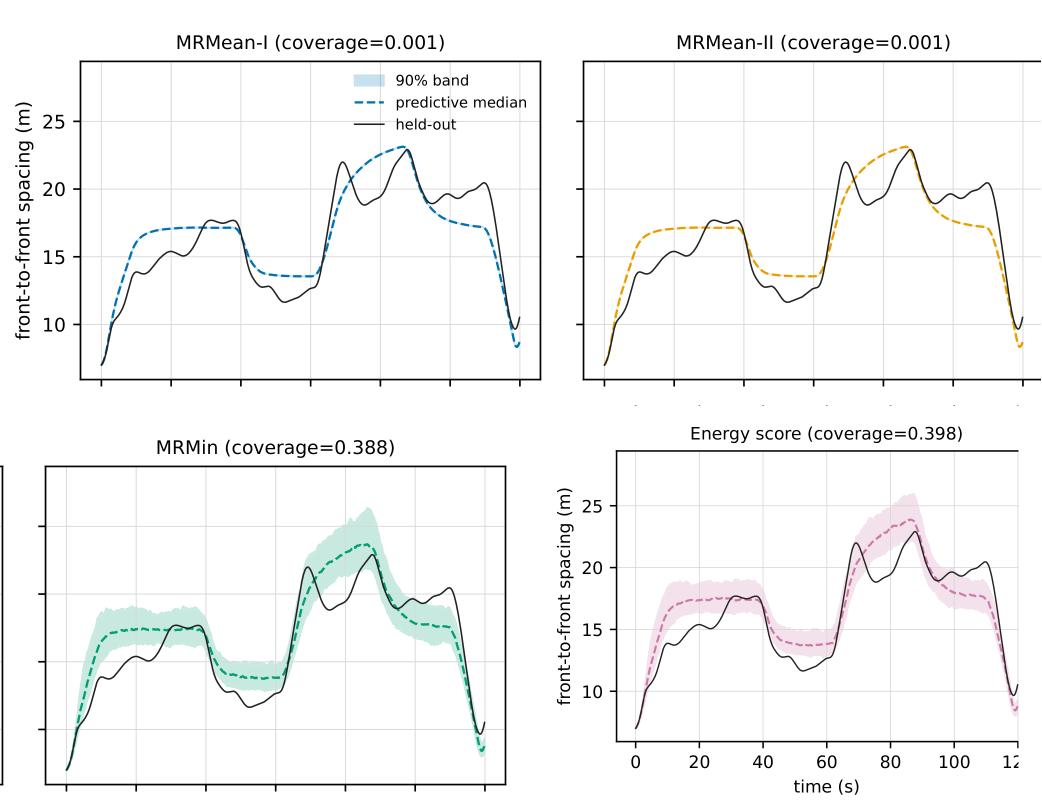}
\caption{Held-out spacing trajectories and $90\%$ predictive bands by calibration
criterion. One panel is shown for MRMean-I, MRMean-II, MRMin, and the energy score;
the solid line is the held-out observation, the central line is the predictive median,
and the coloured envelope is the pointwise predictive band. Driver~11/run~6 is
selected algorithmically because its energy-score coverage, $0.3975$, equals the
median across all 109 folds (ties are broken by lower held-out ES). Its MRMin coverage
is $0.3883$; aggregate results use every fold.}
\label{fig:exp6-bands}
\end{figure}

\begin{table*}[t]
\centering
\caption{Out-of-sample predictive diagnostics from 109 empirical QIDM leave-one-run-out
folds across 11 drivers. Coverage refers to the nominal $90\%$ band,
PIT--KS is the Kolmogorov--Smirnov statistic against uniformity, and std-ratio is
simulated divided by observed across-run standard deviation.}
\label{tab:exp6}
\begin{tabular}{lrrrrr}
\toprule
Criterion & Coverage & PIT--KS & Std-ratio & Energy score & Variogram score \\
\midrule
MRMean-I & 0.001 & 0.762 & 0.000 & 71.67 & 799.68 \\
MRMean-II & 0.017 & 0.751 & 0.019 & 71.15 & 788.94 \\
MRMin & 0.344 & 0.541 & 0.557 & 60.69 & 626.99 \\
Energy score & 0.421 & 0.505 & 0.738 & 60.11 & 610.84 \\
\bottomrule

\end{tabular}
\end{table*}

\subsection{Which features of real variability do the models reproduce?}
\label{sec:exp7}

The earlier analysis separated random run-to-run variation from uncertainty about the
fitted model, but it did not explain the source of the run-to-run variation. Repeated
trajectories let us ask more than which model has the smallest overall error: we can
test whether a model reproduces \emph{when} variability appears and \emph{how long}
a disturbance lasts. We examine three patterns chosen before fitting. Speed-scaled
noise predicts that across-run spacing variability grows
systematically with speed even after controlling for transitions. Action-point or
state-switching mechanisms predict bursts associated with large $|\Delta v|$ or
acceleration and persistence between switches. Additive white acceleration noise
predicts variance accumulated through integration and an aleatoric residual with
little remaining serial correlation once the deterministic response is removed.
The persistent-acceleration reference used below does not mean that the vehicle
continues to accelerate in one direction. It means that the acceleration innovation
left after the deterministic CF response has been removed is serially correlated.
Specifically, its standardised innovation follows
$u_{t+\Delta t}=\rho u_t+\sqrt{1-\rho^2}\epsilon_t$, with $\rho=0.92$ and
$\epsilon_t\sim\mathcal N(0,1)$ independently. A positive or negative innovation
therefore tends to retain its sign and magnitude over nearby updates before decaying.
This AR(1) construction is a diagnostic reference, not a claim that it is the
uniquely identified physical law of driver behaviour.
These signatures can overlap in total variance, but they differ in when that variance
appears and how long residual innovations persist.

The first diagnostic regresses the across-run standard deviation of spacing on speed,
absolute relative speed, and acceleration. The response is computed within driver at
each aligned time point, so the regression concerns repeat-to-repeat variability under
a common stimulus rather than pooled differences between drivers. Driver fixed effects
absorb stable following style, and uncertainty is clustered by driver and time block to
avoid treating the dense 10-Hz observations as independent.
Table~\ref{tab:exp7} reports coefficients on a common scale, their uncertainty,
and the incremental explanatory contribution of each mechanism-linked predictor.
Speed is positive for all 11 drivers and across four filters, but is null within the
middle-speed third. Absolute relative speed is weaker and positive for 8/11 drivers;
acceleration is not significant. Speed scaling is robust across drivers and filters
but not separable from regime everywhere, while the tested contemporaneous
acceleration association is unsupported.

Figure~\ref{fig:exp7-signatures} in the electronic companion forces the regression
result back onto the trajectory, overlaying observed speed-dependent variability with
simulated reference signatures for white-noise QIDM, 2D-IDM, and persistent
acceleration on a common follower-speed coordinate. A coefficient alone is not
sufficient: a speed association can arise because speed indexes a particular segment
of the prescribed trajectory, while an acceleration association can reflect filtering
at regime boundaries. The sensitivity analysis therefore repeats the fit by driver, by
regime, with alternative low-pass settings, and after adding the measured
lead-trajectory deviation as a covariate. Mechanism claims below are limited to
features that survive those checks.

The second diagnostic asks whether an acceleration disturbance disappears immediately
or carries over to later times. We subtract the acceleration predicted by the
deterministic part of the model, scale the remaining error by the model's predicted
noise size, and correlate that residual with itself at later times. White-noise QIDM
predicts substantially faster correlation decay than a persistent-disturbance process;
2D-IDM and action-point mechanisms predict persistence because the latent desired
headway or action state survives across updates. Figure~\ref{fig:exp7-acf} shows
the empirical autocorrelation with driver-level uncertainty beside model-generated
reference envelopes. Disagreement between energy- and variogram-score calibrations
can localise a deficiency that Euclidean energy distance alone may not expose, but
agreement is only a weak diagnostic and does not establish adequate dependence. Across the 44
driver--method cells in the out-of-sample validation, ES and VS have Spearman correlation $0.713$
($p=5.53\times10^{-8}$) and exactly the same ordering: energy score, MRMin,
MRMean-II, then MRMean-I. The two diagnostics therefore agree on ranking even though
the residual ACF is inconsistent with the tested white-noise reference at long lags.

The final diagnostic separates between-driver heterogeneity from within-driver
aleatoric variation. A variance decomposition is performed at each regime and then
aggregated with pre-specified weights. Stable differences in mean headway belong to
between-driver heterogeneity and should not be absorbed into a within-driver noise
parameter. Conversely, repeat-to-repeat deviations under the same driver and stimulus
are the variation a stochastic mechanism must reproduce. The decomposition assigns
$24.7\%$ of residual variance within drivers and $75.3\%$ between drivers. Stable
heterogeneity dominates, but the within-driver component remains substantial.

Together these analyses compare observable implications of candidate mechanisms;
they do not causally identify a unique physical process. At 5 s the
empirical residual ACF is $0.569$, compared with $0.219$ for white QIDM,
$0.152$ for 2D-IDM, and $0.641$ for persistent acceleration: the 2D-IDM
action-point mechanism used in the earlier study
\citep{zhou2025calibration} under-predicts the observed persistence by a
wide margin, a contrast that the repeated-measurement design makes visible
and that a single-trajectory comparison could not have exposed.
Persistent disturbances are most consistent with the observed long-lag pattern;
the tested white and 2D-IDM references under-predict persistence at long lags.
Speed and $|\Delta v|$ signatures do not uniquely identify the physical origin.

\subsection{Parameter uncertainty and interval coverage}
\label{sec:exp8}

The last experiment contrasts two uncertainty diagnostics made available by the
repeated design. The first is the local energy-score sandwich scale
$V_1=A^{-1}BA^{-1}$ from Section~\ref{sec:theory-asymptotics}, evaluated with
$N=32$ common-random-number ensembles. The second is the empirical covariance of
the available published per-run QIDM and 2D-IDM estimates over independent
repetitions. These estimates come from the earlier one-step procedure, not from an
identical energy-score estimator, so the comparison is an external stability
benchmark rather than two estimators of the same covariance matrix. Both quantities
refer to single-run parameter variation; no factor-$n$ covariance of a repeated-run
average is used.

Figure~\ref{fig:exp8-cov} and Table~\ref{tab:exp8} compare their diagonal scale and
correlation structure and separately report curvature eigenvalues. Median ratios of
the regularised local sandwich SD to the empirical per-run SD are $3.57$ and $1.36$,
and off-diagonal correlation RMSEs are $0.717$ and $0.777$. Although the median
driver has no negative curvature direction for either model, the minimum raw
eigenvalue across drivers is $-61.87$ for QIDM and $-294.55$ for 2D-IDM. The
fixed-$N$, fixed-common-random-number sandwich is therefore unstable in these fits
and is not validated as a substitute for independent repetitions.

The coverage experiment reconstructs the single-trajectory overlapping-block scheme
used in \citet{zhou2025calibration}, where block-specific parameter estimates were
treated as independent replicates although most sampled blocks overlap. Overlap is
standard in a moving-block bootstrap and is not, by itself, an error; the issue tested
here is the independence interpretation together with the chosen long block length.
Under known synthetic truths the experiment
compares nominal with empirical interval coverage; the repeated-run estimator supplies an independent
benchmark. This is a methodological comparison enabled by the new design: its purpose
is to quantify the effect of overlap and block length, not to attribute discrepancies
to an optimiser. For $B=100,L=500,T=1200$ and integer starts 1--701, the exact
pairwise overlap probability is $0.917375$. At nominal 95\%, the reconstructed
overlapping-block procedure has coverage $0.800$, versus $0.956$ for 12 independent
repetitions. This known-truth result establishes under-coverage for the reconstructed
procedure; it does not attribute the discrepancy to overlap alone.
Figure~\ref{fig:exp8-coverage} shows the complete nominal-versus-empirical coverage
curves, rather than only the 95\% comparison quoted above.

\section{Conclusions and practical guidance}
\label{sec:conclusions}

This paper establishes a distributional calibration theory for stochastic traffic
simulators. The three common multi-run criteria are not interchangeable: MRMean-I
penalises predictive dispersion and collapses a separable stochastic scale; MRMean-II
scores only the ensemble mean and cannot identify a pure dispersion coordinate; and
MRMin retains the closest trajectory, and the value it targets changes with the
number of simulated trajectories. The Gaussian calculations, full $(K,N)$ grid,
joint QIDM and 2D-IDM recoveries, and controlled exact-MLE experiment reproduce these distinct
signatures. These are score-level results for vector-valued simulator outputs; QIDM
and 2D-IDM are validation vehicles rather than limits on applicability.

The resulting decision rule is simple. If the model can evaluate the probability of
each observed state change, use exact MLE. The controlled QIDM experiment is nearly
unbiased and shows that the earlier failure of an approximate Gaussian residual
likelihood should not be blamed on MLE itself. If the model can simulate complete
trajectories but cannot evaluate their probability, use an unbiased estimator of a
proper score. The energy score
has the same expected objective for every ensemble size, although its random argmin
can still have finite-$N$ error. MRMin may recover planted parameters at a favourable
fixed $N$, as \citet{zhou2025calibration} reported, but that numerical success does
not give it an ensemble-size-invariant population minimiser.

Repeated measurements separate calibration quality from model structure. Energy-score
calibration gives the best relative distributional prediction among the four tested
objectives, and calibration on multiple observed trajectories reduces the QIDM
discrepancy. Nevertheless, the nominal 90\% band covers only $42.1\%$ of held-out
trajectories in the 109-fold comparison. The selected QIDM and 2D-IDM specifications
also underrepresent long-lived disturbances. Properness therefore prevents the
calibration objective from preferring a predictive distribution that is too narrow
merely because its spread was not rewarded; it does not guarantee that an inadequate
model family can reproduce every
parameter or dependence feature. Energy-score and variogram diagnostics should be
reported together when temporal dependence is scientifically relevant.

Uncertainty quantification requires the same separation. For the exact score, the
sandwich formula says that covariance decreases as $1/n$ when $n$ independent
trajectories are averaged. A curvature calculation using finitely many simulations
only describes the local shape of that particular numerical objective. Some fitted
objectives are not locally bowl-shaped, and intervals constructed from overlapping
blocks cover the true value less often than their nominal level;
independent repeated trajectories should therefore be used directly when available.
The present data are limited to one site, vehicle pair, and prescribed stimulus, and
MRMin monotonicity remains conditional on decreasing differences for a new nonlinear
simulator. Future work should extend the repeated-measurement design across sites and
develop dependence-sensitive, computationally efficient proper scores. The practical
contribution is thus not a claim that one criterion always recovers every parameter
best, but a framework that distinguishes an invalid target, weak identification,
finite-sample error, and structural model misspecification.

\section*{Declaration of competing interest}
The authors declare no known competing financial interests or personal relationships
that could have appeared to influence the work reported in this paper.

\section*{Data availability}
The processed trajectory data supporting the findings are available from the
corresponding author upon reasonable request. A replication package containing
analysis code, random seeds, checkpoint files, and figure-level numerical data will be
deposited in a public repository upon publication, subject to the terms governing the
source data.
\bibliographystyle{plainnat}
\bibliography{refs}

\ECSwitch
\ECHead{Electronic Companion: Proofs and additional results}
\section{Complete proofs}
\subsection{Proof of Theorem~\ref{thm:mrmean1}}
\begin{proof}
Write
\[
  X - y = \big(X - \mu_\beta\big) - \big(y - \mu_0\big) + \big(\mu_\beta - \mu_0\big).
\]
The three terms on the right have, respectively, expectation zero conditional on
nothing (the first two are centred by construction) and are pairwise uncorrelated:
$X - \mu_\beta$ is independent of $y - \mu_0$ by Assumption~\ref{assump:standing}(b)
and both are independent of the constant vector $\mu_\beta - \mu_0$. Expanding the
squared norm,
\[
\begin{aligned}
  \|X-y\|^2
  ={}& \|X-\mu_\beta\|^2+\|y-\mu_0\|^2
       +\|\mu_\beta-\mu_0\|^2 \\
     &+2\langle X-\mu_\beta,\mu_\beta-\mu_0\rangle \\
     &-2\langle X-\mu_\beta,y-\mu_0\rangle \\
     &-2\langle y-\mu_0,\mu_\beta-\mu_0\rangle .
\end{aligned}
\]
and every cross term has zero expectation under $\E_X \E_y$: the first cross term
because $\E_X(X - \mu_\beta) = 0$, the second because $\E_y(y-\mu_0) = 0$. Taking
$\E_X$ first and then $\E_y$ of what remains gives
\[
\begin{aligned}
  \E_y\E_X\|X-y\|^2
  ={}&\E_X\|X-\mu_\beta\|^2
      +\E_y\|y-\mu_0\|^2 \\
     &+\|\mu_\beta-\mu_0\|^2 \\
  ={}&\tr\Sigma_\beta+\tr\Sigma_0
      +\|\mu_\beta-\mu_0\|^2 .
\end{aligned}
\]
which is \eqref{eq:s1-decomp}. Since $\tr\Sigma_0$ does not depend on $\beta$, the
population objective is, up to an additive constant, $\|\mu_\beta-\mu_0\|^2 +
\tr\Sigma_\beta$. Under separability the mean term can be held at its minimum value
$0$ by fixing $\mu_\beta = \mu_0$ while $\Sigma_\beta$ still ranges over an interval
extending down to values arbitrarily close to $0$; on that slice the objective is
strictly increasing in $\tr\Sigma_\beta$, so its infimum is attained (or approached)
at the lower end of the admissible range for $\Sigma_\beta$. Hence any minimising
sequence has $\tr\Sigma_\beta \to 0$. Because the true $\tr\Sigma_0 > 0$ in general,
the population minimiser does not coincide with $\beta_0$, so $S_1$ is not proper.
\end{proof}

\subsection{Proof of Theorem~\ref{thm:mrmean2}}
\begin{proof}
Write $\mu_\beta-y=(\mu_\beta-\mu_0)-(y-\mu_0)$.  The cross term has zero
expectation by centring, while
$\E\|y-\mu_0\|^2=\tr\Sigma_0$.  This gives
\eqref{eq:s2-decomp}, whose right-hand side contains no $\Sigma_\beta$.
\end{proof}

\subsection{Proof of Theorem~\ref{thm:mrmin}}
\begin{proof}
Let $u_N=\inf\arg\min h_N$ and suppose, to the contrary, that
$u_{N+1}<u_N$.  Optimality gives
\[
 h_N(u_N)\leq h_N(u_{N+1}),\qquad
 h_{N+1}(u_{N+1})\leq h_{N+1}(u_N).
\]
Subtracting implies
\[
 h_{N+1}(u_{N+1})-h_N(u_{N+1})
 \leq h_{N+1}(u_N)-h_N(u_N),
\]
which contradicts strict decrease because $u_{N+1}<u_N$.  Hence
$u_{N+1}\geq u_N$, and exponentiation preserves the order.  This is the
decreasing-differences form of the Topkis argument.  The numerical check of its
premise is reported in Section~\ref{sec:exp1}.
\end{proof}

\subsection{Proof of Proposition~\ref{prop:mrmin-dimension}}
\begin{proof}
Let $R_N(y)=\min_{i\le N}\|X_i-y\|$. Conditional on $Y=y$,
\[
 \Pr\{R_N(y)>r\mid y\}
 =\big[1-P_\sigma\{B(y,r)\}\big]^N.
\]
The stated local expansion therefore gives
$Nf_\sigma(y)V_KR_N(y)^K\Rightarrow\operatorname{Exp}(1)$. Moment convergence,
under the stated uniform-integrability condition, yields
\[
 \E\{R_N(y)^2\mid y\}\sim
 \frac{\Gamma(1+2/K)}{\{NV_Kf_\sigma(y)\}^{2/K}}.
\]
Averaging over $Y$ gives \eqref{eq:mrmin-nn-asymptotic}. For the Gaussian toy,
\[
 f_\sigma(y)=(2\pi\sigma^2)^{-K/2}
 \exp\{-\|y-\mu_0\|^2/(2\sigma^2)\},
\]
Gaussian integration yields
\[
 \E\{f_\sigma(Y)^{-2/K}\}
 =2\pi\sigma^2
 \left(1-\frac{2\sigma_0^2}{K\sigma^2}\right)^{-K/2}.
\]
Differentiating its logarithm with respect to $\sigma^2$ gives the unique
stationary point $\sigma^2=\sigma_0^2(1+2/K)$, which is its minimum.
\end{proof}

\subsection{Proof of Proposition~\ref{prop:es-unbiased}}
\begin{proof}
For the first term, each $X_i \sim P$, so $\E\|X_i - y\|= \E\|X-y\|$ termwise
and averaging preserves this expectation. For the second term, the sum
$\sum_{i \ne j}$ ranges over exactly $N(N-1)$ ordered pairs with $i \ne j$; by
independence and identical distribution of the $X_i$, every such pair satisfies
$\E\|X_i - X_j\| = \E\|X - X'\|$ for independent $X, X' \sim P$. Hence
\[
\begin{aligned}
 &\E\left[\frac{1}{2N(N-1)}
       \sum_{i\ne j}\|X_i-X_j\|\right]\\
 &\quad=\frac{N(N-1)}{2N(N-1)}\,
       \E\|X-X'\|\\
 &\quad=\tfrac12\,\E\|X-X'\| .
\end{aligned}
\]
Subtracting gives $\E\,\widehat{\mathrm{ES}}_N = \E\|X-y\| - \tfrac12\E\|X-X'\| =
\mathrm{ES}(P,y)$, for every $N \ge 2$.
\end{proof}

\section{Calibration literature and experimental design}
\begin{table*}[t]
\centering
\caption{Conceptual progression from \citet{zhou2025calibration} to the present study.}
\label{tab:progression}
\small
\begin{tabular}{@{}>{\raggedright\arraybackslash}p{0.16\linewidth}
                    >{\raggedright\arraybackslash}p{0.35\linewidth}
                    >{\raggedright\arraybackslash}p{0.39\linewidth}@{}}
\toprule
Aspect & Zhou et al. (2025) & Present study \\
\midrule
Evaluation principle & Parameter recovery in finite synthetic designs & Strict propriety and the long-run target of each criterion \\
MRMean & Run-averaged errors suppress stochasticity & MRMean-I collapses separable dispersion; MRMean-II leaves pure dispersion unidentified \\
MRMin & Best-of-$N$ method recommended after accurate fixed-$N$ grid recovery & Nearest-neighbour rule with an $N$-indexed target; monotonicity is conditional on decreasing differences \\
Random seed & Implicit variable selected in a nested optimisation & Monte Carlo draw used to integrate over $P_\beta$, not part of the estimand \\
Likelihood & Gaussian trajectory-residual construction for QIDM/2D-IDM, alongside a separate model-specific analytic likelihood & Approximate residual likelihood separated from a correctly specified QIDM transition likelihood \\
Synthetic design & Truth selected from a discretised candidate space and recovered by traversal & Off-grid truth, continuous joint calibration, multiple regimes, and observation error \\
Uncertainty & Overlapping blocks sampled from one trajectory & Independent repetitions, M-estimation comparison, and known-truth coverage \\
Validation & Error summaries and prediction profiles & Held-out proper scores, coverage, PIT, dependence, and mechanism diagnostics \\
Recommendation & Adopt MRMin to preserve stochasticity & Use an exact likelihood when available; otherwise use an unbiased proper score \\
\bottomrule
\end{tabular}
\end{table*}

\begin{table}[t]
\centering
\caption{Multi-run calibration practice for stochastic CF models, reclassified
by the scoring rule each procedure implicitly applies to the predictive
distribution $P_\beta$, rather than by measure of performance or optimiser.
$N$ denotes the number of simulated trajectories per calibration evaluation.}
\label{tab:rw-implicit-scores}
\begin{tabular}{p{0.24\linewidth}p{0.68\linewidth}}
\toprule
Implicit scoring rule & Representative use in the CF literature \\
\midrule
Mean of run-wise error (MRMean-I) &
Multiple simulated runs per parameter draw, each compared with the
observation and averaged; the form analysed as $S_1$ in
Section~\ref{sec:mrmean1} and used as one of the two MRMean variants in
\citet{zhou2025calibration}. \\
Error of the run-averaged trajectory (MRMean-II) &
Multiple simulated runs averaged into a single mean trajectory before
comparison with the observation; the second MRMean variant in
\citet{zhou2025calibration}, analysed as $S_2$ in Section~\ref{sec:mrmean2}. \\
Best of $N$ (MRMin) &
Only the closest of $N$ simulated trajectories is retained; introduced and
evaluated in \citet{zhou2025calibration}, analysed as $S_3^{(N)}$ in
Section~\ref{sec:mrmin} and identified there with the ABC acceptance rule of
Section~\ref{sec:rw-sbi}. \\
Likelihood under an assumed residual law &
A parametric error distribution (e.g.\ Gaussian residuals with a
data-source- or sampling-interval-dependent variance) is assumed and
maximised, as in the local maximum-likelihood calibration of
\citet{kesting2008calibrating} and the multi-source calibration of
\citet{hoogendoorn2010calibration}. \\
Exact transition likelihood &
The model provides a formula for the probability density of the next observed
state given the current state and parameters. This density can therefore be
evaluated directly rather than approximated from simulated runs, as in the
two-regime stochastic Newell inference of
\citet{xu2020statistical} and the Langevin formulation of
\citet{ngoduy2019langevin}. \\
\bottomrule
\end{tabular}
\end{table}

\section{Additional synthetic results}

\begin{table*}[t]
\centering
\scriptsize
\caption{Numerical verification of the analytic results.}
\label{tab:theory-verification}
\begin{tabular}{p{0.15\linewidth}p{0.21\linewidth}p{0.42\linewidth}p{0.09\linewidth}}
\toprule
Claim & Setting & Numerical result & Verdict \\
\midrule
Theorem~\ref{thm:mrmean1}
& $K=10$, 100 scales in $[10^{-4},3]$
& Objective strictly increases with $\sigma$ and approaches $K\sigma_0^2$ at zero.
& Supported \\
Theorem~\ref{thm:mrmean2}
& $K=10$, 101 scales in $[0,5]$
& Regression slope is zero to machine precision; objective range $<10^{-12}$.
& Supported \\
Mean run-wise RMSE
& $K\in\{1,10,100,1200\}$, 121 scales in $[0,2.5]$
& $c_K\sqrt{\sigma^2+\sigma_0^2}$ is strictly increasing and is minimised at
$\sigma=0$ for every $K$; Monte Carlo agrees within $8\times10^{-3}$.
& Supported \\
Theorem~\ref{thm:mrmin}, decreasing differences
& $K\in\{1,10,100,1200\}$; four adjacent $N$ pairs; 120 positive scales
& All adjacent $u$-differences are negative for $K\geq10$. At $K=1$, the
fractions are $0.815$, $0.487$, $0.412$, and $0.378$ as $N$ increases.
& Conditional; mixed grid \\
MRMin large-$N$ limit
& $K=1$, extended through $N=10^6$
& The minimiser reaches $1.731$--$1.732$ by $N=10^4$--$10^5$ and the large-$N$ value is $\sqrt3=1.73205$.
& Finite saturation \\
High-dimensional finite-budget approximation
& $K\in\{1,10,100,1200\}$, $N=10,\ldots,2000$
& $\sqrt{\log N}$ has the largest $R^2$ for every $K$; full comparison is given below.
& Empirically supported in range \\
Energy-score propriety
& $K\in\{1,2,10,100\}$
& Golden-section minimisers equal $\sigma_0$ within $2\times10^{-5}$; direct Monte Carlo agrees within $8\times10^{-3}$.
& Supported \\
ES with $1/[N(N-1)]$
& 35 ensemble sizes, $N=5,\ldots,2000$
& Expected slope against $\log N$ and objective range are zero to machine precision.
& Unbiased \\
ES with $1/N^2$
& Same ensemble sizes
& Log-$N$ slope is nonzero ($p<10^{-6}$); bias times $N$ is constant to machine precision.
& $O(N^{-1})$ bias \\
\bottomrule
\end{tabular}
\end{table*}

\begin{figure}[t]
\centering
\includegraphics[width=0.78\linewidth]{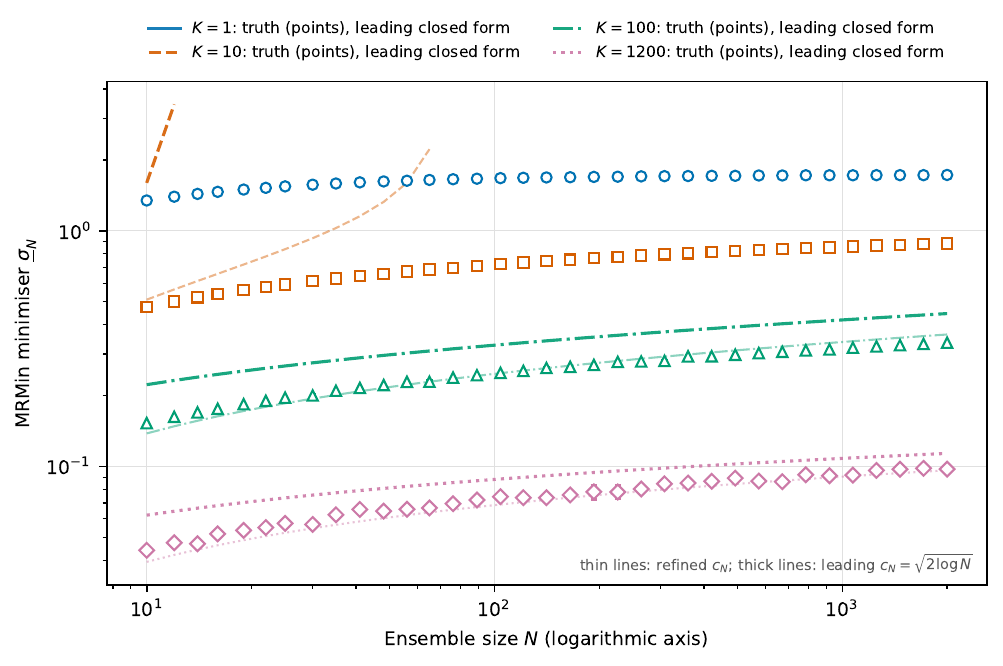}
\caption{MRMin minimiser $\underline\sigma_N$ against ensemble size on a
logarithmic $N$ axis for $K\in\{1,10,100,1200\}$. Points are numerical
minimisers; lines are the fitted $a+b\sqrt{\log N}$ forms, with finite-range
log--log exponent $\widehat\alpha$ annotated. Both axes are logarithmic. Thin
short segments show the historical $N^{1/K}$ reference; the fitted curves are
interpreted as pre-asymptotic because the corrected theory has a finite
large-$N$ limit.}
\label{fig:exp1-dimension}
\end{figure}

\begin{table*}[t]
\centering
\caption{Distribution of calibrated stochastic parameters in the
$N$-sensitivity experiment. Spearman coefficients are calculated from the
seven median estimates across $N$ and each coefficient is shown once in a cell
spanning the corresponding seven $N$ rows; ``--'' denotes an undefined coefficient
when all medians are identical. Median, IQR, and bias are printed to four decimals:
the small 2D-IDM changes with $N$ are real but were hidden by three-decimal rounding,
whereas QIDM MRMean-I is exactly constant because every median is at the zero bound.
Model and criterion labels are likewise merged over their corresponding row blocks.}
\label{tab:exp2}
\scriptsize
{\renewcommand{\arraystretch}{0.94}
\begin{tabular}{lllrrrr}
\toprule
Model & Criterion & $N$ & Median & IQR & Bias & Spearman $\rho$ \\
\midrule
\multirow{28}{*}{2D-IDM} & \multirow{7}{*}{Energy score} & 20 & 0.7282 & 0.0166 & -0.0018 & \multirow{7}{*}{0.321} \\
 &  & 50 & 0.7279 & 0.0173 & -0.0021 &  \\
 &  & 100 & 0.7278 & 0.0169 & -0.0022 &  \\
 &  & 200 & 0.7276 & 0.0173 & -0.0024 &  \\
 &  & 500 & 0.7281 & 0.0175 & -0.0019 &  \\
 &  & 1000 & 0.7282 & 0.0169 & -0.0018 &  \\
 &  & 2000 & 0.7283 & 0.0175 & -0.0017 &  \\
 & \multirow{7}{*}{MRMean-I} & 20 & 0.7254 & 0.0166 & -0.0046 & \multirow{7}{*}{0.464} \\
 &  & 50 & 0.7251 & 0.0173 & -0.0049 &  \\
 &  & 100 & 0.7251 & 0.0168 & -0.0049 &  \\
 &  & 200 & 0.7248 & 0.0173 & -0.0052 &  \\
 &  & 500 & 0.7252 & 0.0175 & -0.0048 &  \\
 &  & 1000 & 0.7256 & 0.0171 & -0.0044 &  \\
 &  & 2000 & 0.7254 & 0.0176 & -0.0046 &  \\
 & \multirow{7}{*}{MRMean-II} & 20 & 0.7279 & 0.0156 & -0.0021 & \multirow{7}{*}{0.536} \\
 &  & 50 & 0.7281 & 0.0169 & -0.0019 &  \\
 &  & 100 & 0.7279 & 0.0159 & -0.0021 &  \\
 &  & 200 & 0.7275 & 0.0163 & -0.0025 &  \\
 &  & 500 & 0.7280 & 0.0165 & -0.0020 &  \\
 &  & 1000 & 0.7283 & 0.0161 & -0.0017 &  \\
 &  & 2000 & 0.7282 & 0.0166 & -0.0018 &  \\
 & \multirow{7}{*}{MRMin} & 20 & 0.7267 & 0.0230 & -0.0033 & \multirow{7}{*}{0.071} \\
 &  & 50 & 0.7263 & 0.0184 & -0.0037 &  \\
 &  & 100 & 0.7253 & 0.0206 & -0.0047 &  \\
 &  & 200 & 0.7243 & 0.0201 & -0.0057 &  \\
 &  & 500 & 0.7254 & 0.0259 & -0.0046 &  \\
 &  & 1000 & 0.7258 & 0.0247 & -0.0042 &  \\
 &  & 2000 & 0.7267 & 0.0238 & -0.0033 &  \\
\multirow{28}{*}{QIDM} & \multirow{7}{*}{Energy score} & 20 & 0.3559 & 0.2377 & -0.1141 & \multirow{7}{*}{0.714} \\
 &  & 50 & 0.3531 & 0.2532 & -0.1169 &  \\
 &  & 100 & 0.3585 & 0.2140 & -0.1115 &  \\
 &  & 200 & 0.3744 & 0.2110 & -0.0956 &  \\
 &  & 500 & 0.3821 & 0.2209 & -0.0879 &  \\
 &  & 1000 & 0.3768 & 0.2276 & -0.0932 &  \\
 &  & 2000 & 0.3692 & 0.2181 & -0.1008 &  \\
 & \multirow{7}{*}{MRMean-I} & 20 & 0.0000 & 0.0000 & -0.4700 & \multirow{7}{*}{--} \\
 &  & 50 & 0.0000 & 0.0000 & -0.4700 &  \\
 &  & 100 & 0.0000 & 0.0000 & -0.4700 &  \\
 &  & 200 & 0.0000 & 0.0000 & -0.4700 &  \\
 &  & 500 & 0.0000 & 0.0000 & -0.4700 &  \\
 &  & 1000 & 0.0000 & 0.0000 & -0.4700 &  \\
 &  & 2000 & 0.0000 & 0.0000 & -0.4700 &  \\
 & \multirow{7}{*}{MRMean-II} & 20 & 0.1065 & 0.4944 & -0.3635 & \multirow{7}{*}{0.893} \\
 &  & 50 & 0.2305 & 0.5190 & -0.2395 &  \\
 &  & 100 & 0.2458 & 0.4693 & -0.2242 &  \\
 &  & 200 & 0.2644 & 0.4664 & -0.2056 &  \\
 &  & 500 & 0.2751 & 0.4855 & -0.1949 &  \\
 &  & 1000 & 0.2706 & 0.5034 & -0.1994 &  \\
 &  & 2000 & 0.2709 & 0.5053 & -0.1991 &  \\
 & \multirow{7}{*}{MRMin} & 20 & 0.1294 & 0.1135 & -0.3406 & \multirow{7}{*}{1.000} \\
 &  & 50 & 0.1386 & 0.1182 & -0.3314 &  \\
 &  & 100 & 0.1621 & 0.1288 & -0.3079 &  \\
 &  & 200 & 0.1807 & 0.1640 & -0.2893 &  \\
 &  & 500 & 0.1845 & 0.1683 & -0.2855 &  \\
 &  & 1000 & 0.2301 & 0.2103 & -0.2399 &  \\
 &  & 2000 & 0.2378 & 0.2102 & -0.2322 &  \\
\bottomrule
\end{tabular}

}%
\end{table*}

\begin{figure}[t]
\centering
\includegraphics[width=0.86\linewidth]{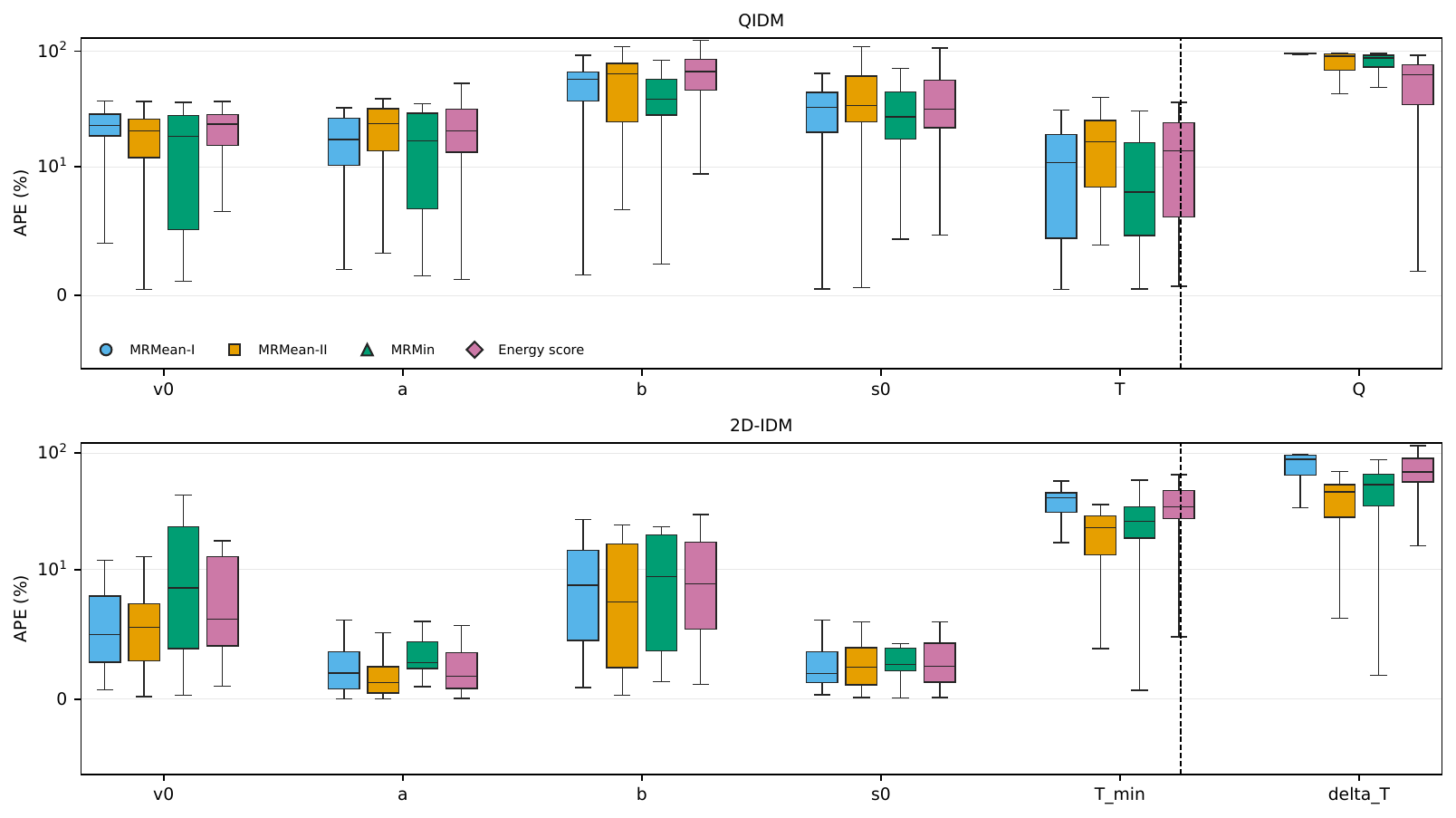}
\caption{Absolute percentage error in joint parameter recovery by parameter,
criterion, and stochastic CF model. Boxes summarise independent off-grid
synthetic calibrations; stochastic coordinates are separated visually from
deterministic coordinates. Panels are QIDM and 2D-IDM; each box
contains ten independent calibrations per test.}
\label{fig:exp3}
\end{figure}

\begin{table*}[t]
\centering
\caption{Joint parameter recovery under three complementary summaries.
Deterministic AAPE averages APE over the five deterministic coordinates;
stochastic APE reports the single stochastic coordinate ($Q$ or $\Delta T$);
range MAE averages absolute errors after division by each prespecified parameter
range. Values average ten independent calibrations. Lower is better.}
\label{tab:exp3}
\scriptsize
\begin{tabular}{llrrrrrr}
\toprule
& & \multicolumn{3}{c}{Test I} & \multicolumn{3}{c}{Test II} \\
\cmidrule(lr){3-5}\cmidrule(lr){6-8}
Model & Criterion & Det. AAPE & Stoch. APE & Range MAE
& Det. AAPE & Stoch. APE & Range MAE \\
\midrule
\multirow{4}{*}{QIDM}
& MRMean-I & 30.5 & 95.0 & 20.6 & 29.8 & 95.2 & 20.0 \\
& MRMean-II & 30.4 & 75.6 & 19.7 & 33.1 & 98.8 & 21.1 \\
& MRMin & 25.7 & 81.5 & \textbf{17.4} & \textbf{26.8} & 79.0 & \textbf{17.5} \\
& Energy score & 33.1 & \textbf{44.4} & 20.0 & 33.6 & \textbf{65.1} & 20.1 \\
\midrule
\multirow{4}{*}{2D-IDM}
& MRMean-I & 12.7 & 77.3 & 12.9 & 16.2 & 81.9 & 14.8 \\
& MRMean-II & \textbf{8.0} & \textbf{35.2} & \textbf{7.5} & \textbf{9.6} & 48.1 & \textbf{9.1} \\
& MRMin & 12.8 & 53.3 & 11.9 & 12.4 & \textbf{45.6} & 10.9 \\
& Energy score & 13.8 & 75.6 & 13.9 & 13.9 & 66.9 & 12.7 \\
\bottomrule
\end{tabular}
\end{table*}

\section{Additional repeated-measurement diagnostics}
\begin{figure}[t]
\centering
\includegraphics[width=0.78\linewidth]{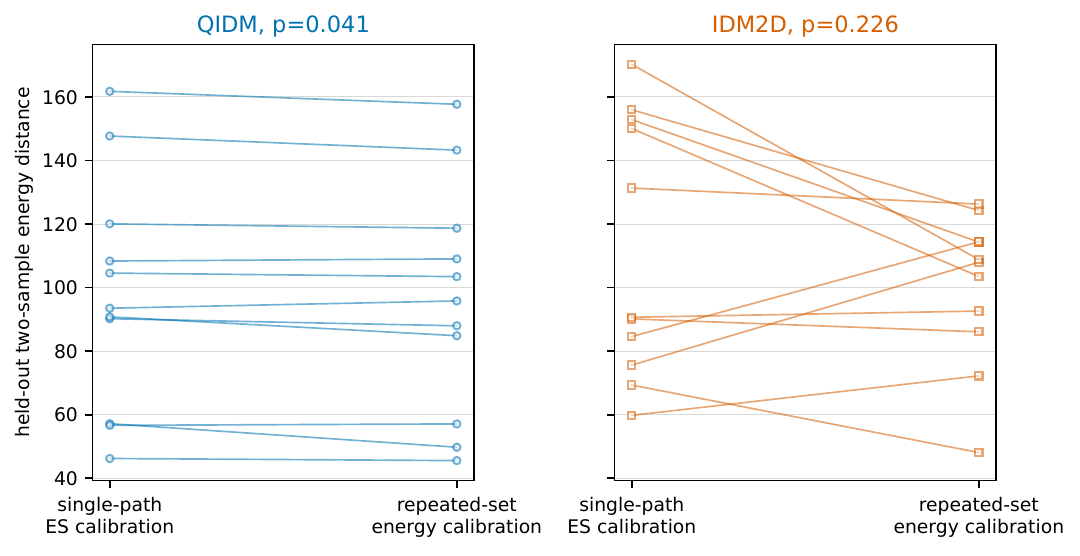}
\caption{Two-sample energy distance between simulated and observed trajectory
ensembles by driver and stochastic CF model. Each point compares held-out observed
repetitions with a newly simulated ensemble; paired coloured symbols contrast
single-path ES calibration with repeated-set energy-distance calibration. The last
two available runs per driver are held out, and uncertainty is assessed by exact
paired sign permutation across 11 drivers.}
\label{fig:exp5}
\end{figure}

\begin{table}[t]
\centering
\caption{Regression of across-run spacing standard deviation on state and action
variables. Coefficients are standardised; two-way clustering uses driver and 5-s time
blocks. Incremental $R^2$ removes one predictor at a time.}
\label{tab:exp7}
\begin{tabular}{lrrrr}
\toprule
Predictor & Coefficient & Standard error & $p$-value & $\Delta R^2$ \\
\midrule
Follower speed $v$ & 0.472 & 0.082 & 8.32e-09 & 0.117 \\
Absolute relative speed $|\Delta v|$ & 0.147 & 0.053 & 0.00587 & 0.020 \\
Acceleration $a$ & -0.038 & 0.043 & 0.377 & 0.001 \\
Lead-trajectory deviation & 0.252 & 0.108 & 0.0191 & 0.034 \\
\bottomrule

\end{tabular}
\end{table}

\begin{figure}[t]
\centering
\includegraphics[width=0.82\linewidth]{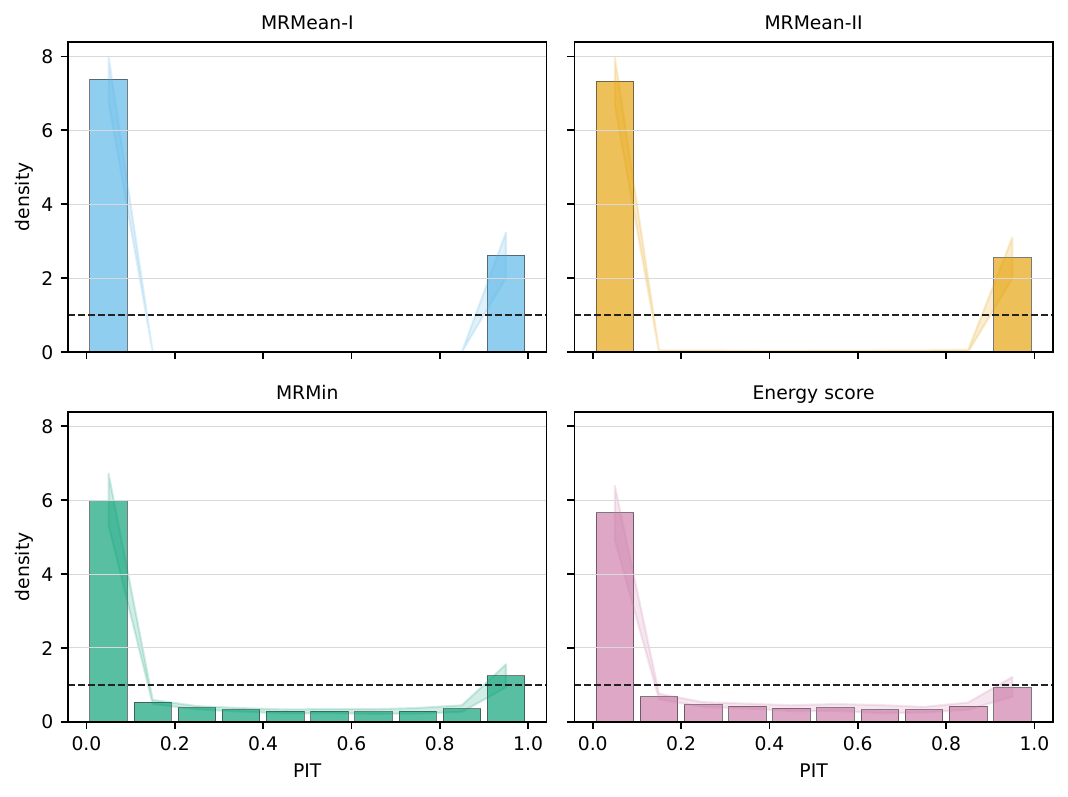}
\caption{Probability integral transform (PIT) histograms for held-out spacing under
each calibration criterion. The horizontal reference denotes uniformity. Ten
equal-width bins are used and bootstrap envelopes resample whole drivers rather than
time points or folds.}
\label{fig:exp6-pit}
\end{figure}

\begin{figure}[t]
\centering
\includegraphics[width=0.68\linewidth]{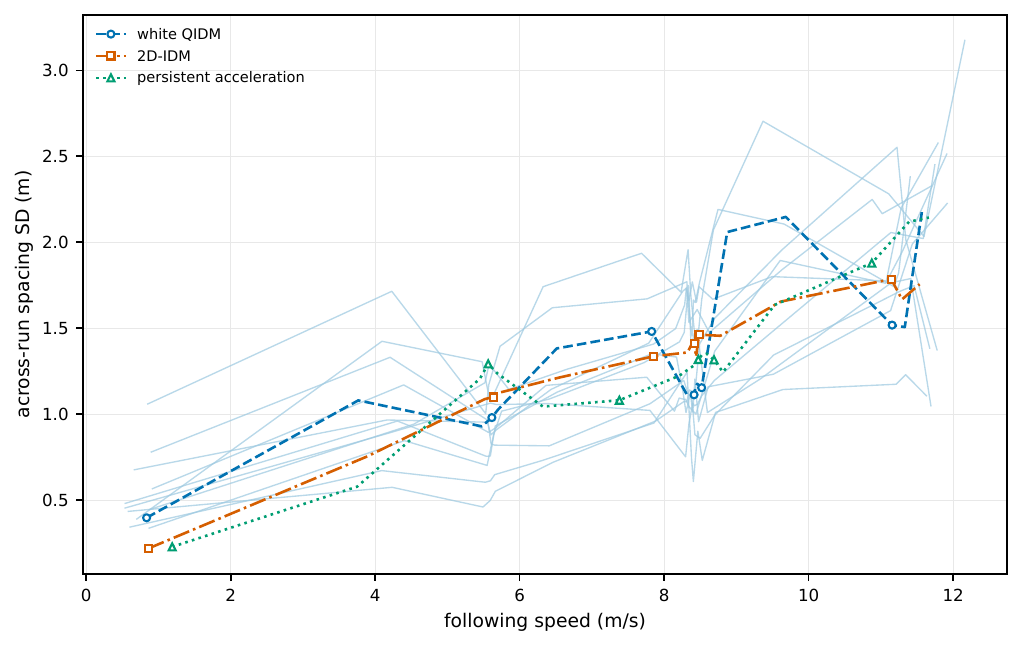}
\caption{Across-run standard deviation of spacing against follower speed, with
overlays for three candidate stochastic-mechanism references: white-noise QIDM,
2D-IDM, and an AR(1) persistent-acceleration innovation with $\rho=0.92$ per
0.1-s update. ``Persistent acceleration'' means persistence of the positive or
negative acceleration innovation, not continuous acceleration of the vehicle.
Driver-level empirical curves are shown in light blue. Simulated follower speed is
reconstructed from simulated position, and each simulated spacing-SD curve is
vertically normalised to the empirical mean SD; colours therefore compare shape
rather than absolute variance. Main curves use raw
10-Hz alignment; widths 5, 11, and 21 samples and speed tertiles are sensitivity checks.}
\label{fig:exp7-signatures}
\end{figure}

\begin{figure}[t]
\centering
\includegraphics[width=0.78\linewidth]{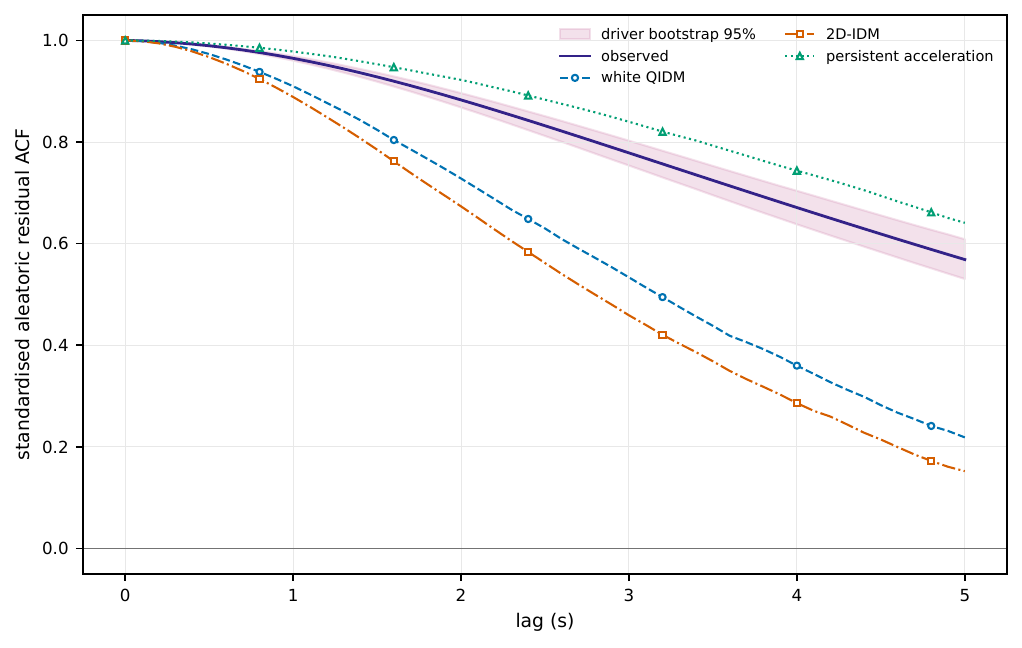}
\caption{How long unexplained acceleration persists. The vertical axis is the
autocorrelation of the standardised within-driver acceleration residual at each lag.
The solid empirical estimate is accompanied by driver-level uncertainty; reference
envelopes are generated under fitted white-noise, 2D-IDM, and action-persistence
mechanisms. Residual correlation outside an envelope identifies temporal structure
not represented by that mechanism. Residuals are deviations from the within-driver
mean spacing trajectory, lags span 0--5 s, and bands resample whole drivers.}
\label{fig:exp7-acf}
\end{figure}

\begin{figure}[t]
\centering
\includegraphics[width=0.82\linewidth]{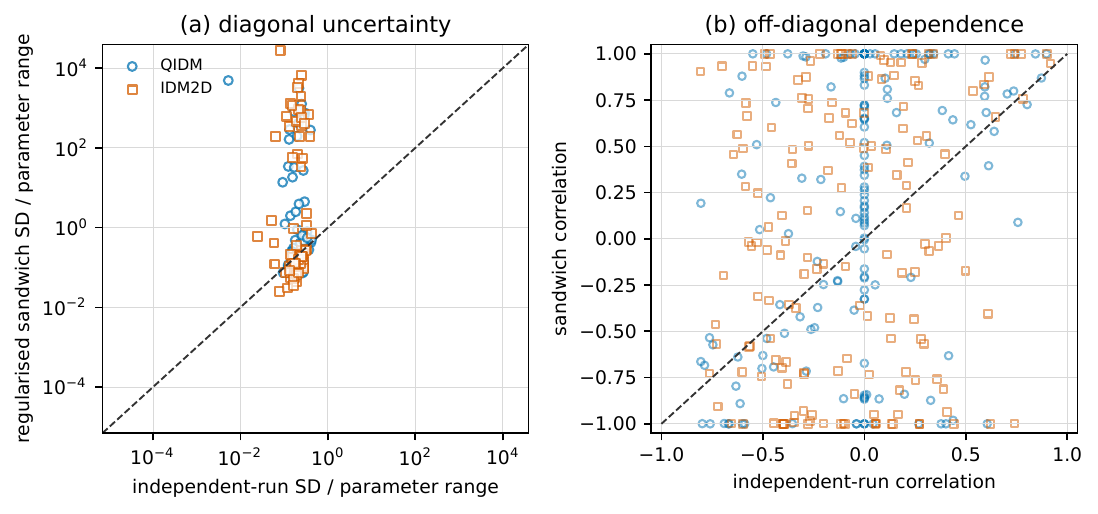}
\caption{Local energy-score sandwich scale and empirical parameter stability across
independent repeated runs. The two routes use different estimators and are therefore
an external diagnostic comparison, not interchangeable estimates of one covariance
matrix. Diagonal panels compare single-run standard-deviation scales; off-diagonal
panels compare correlations, with a $45^\circ$ reference line. Coloured symbols
distinguish QIDM and 2D-IDM across all 11 drivers; regularised sandwich values are
displayed only to diagnose the local approximation.}
\label{fig:exp8-cov}
\end{figure}

\begin{figure}[t]
\centering
\includegraphics[width=0.76\linewidth]{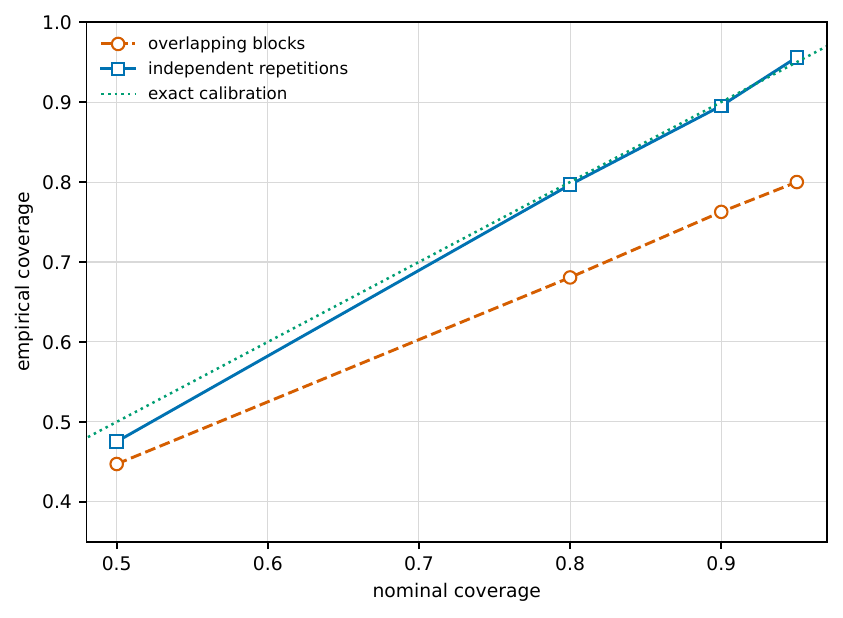}
\caption{Empirical coverage of nominal confidence intervals under the
single-trajectory overlapping-block scheme and the independent-repetition scheme.
Coverage is plotted against nominal level; the diagonal denotes exact calibration and
colours and symbols distinguish the two schemes. The simulation uses 1500 data sets,
$B=100$, $L=500$, $T=1200$, AR(1) correlation $0.92$, and 12 independent series.}
\label{fig:exp8-coverage}
\end{figure}

\begin{table*}[t]
\centering
\caption{Real repeated-run score-curvature diagnosis. The SD ratio compares the
regularised, fixed-common-random-number energy-score sandwich scale with the empirical
SD of published per-run estimates from a different procedure; it is a stability
benchmark, not a covariance calibration ratio. Correlation RMSE compares all
off-diagonal entries. Negative raw eigenvalues invalidate the local interpretation.}
\label{tab:exp8}
\begin{tabular}{lrrrrr}
\toprule
Model & Drivers & Median negative eig. & Minimum eig. & Correlation RMSE & Median SD ratio \\
\midrule
QIDM & 11 & 0.0 & -61.87 & 0.717 & 3.57 \\
2D-IDM & 11 & 0.0 & -294.55 & 0.777 & 1.36 \\
\bottomrule

\end{tabular}
\end{table*}

\end{document}